\documentclass[journal,10pt]{IEEEtran}
\normalsize
\usepackage{dsfont}
\usepackage{amssymb, amsmath, epsfig, booktabs, amsthm, epstopdf, color,caption, import, cite}
\usepackage{graphicx}
\usepackage{subfigure}
\usepackage{psfrag}
\usepackage{paralist}
\usepackage{float}
\usepackage[margin=0.65in]{geometry}
\allowdisplaybreaks[4]

\newtheorem{theorem}{Theorem}

\newtheorem{lemma}{Lemma}

\newtheorem{definition}{Definition}
\newtheorem{remark}{Remark}

\renewcommand{\footnoterule}{%
  \kern -3pt
  \hrule width 150pt height .2pt
  \kern 2pt
}

\begin{document}

\title{On Lossy Joint Source-Channel Coding In Energy Harvesting Communication Systems}
\author{\normalsize Meysam~Shahrbaf Motlagh,\IEEEmembership{} Masoud~Badiei Khuzani,\IEEEmembership{} Patrick~Mitran~\IEEEmembership{}

\thanks{Meysam Shahrbaf Motlagh and Patrick Mitran are with the Department of Electrical and Computer Engineering, University of Waterloo, Waterloo, ON, Canada. Email: $\lbrace$mshahrba, pmitran$\rbrace$@uwaterloo.ca. \newline Masoud Badiei Khuzani is with the School of Engineering and Applied Science, Harvard University, MA, USA. Email: mbadieik@seas.harvard.edu. \newline
This paper was presented in part at the International Symposium on Information Theory (ISIT), Honolulu, Hawaii, USA, 2014 \cite{motlagh2014lossy}. \newline
This work was supported by the Natural Sciences and Engineering Research Council of Canada under Grant RGPIN 353447-2013.}}
\markboth{IEEE Transactions on Communications}{}
\maketitle

\begin{abstract}
We study the problem of lossy joint source-channel coding in a single-user energy harvesting communication system
with causal energy arrivals and the energy storage unit may have leakage. In particular, we investigate the achievable distortion in the transmission of a single source via an energy harvesting transmitter over a point-to-point channel. We consider an adaptive joint source-channel coding system, where the length of channel codewords 
varies based on the available battery charge. We first establish a lower bound on the achievable distortion. Then, as necessary conditions for local optimality, we obtain two coupled equations that determine the mismatch ratio between channel symbols and source symbols as well as the transmission power, both as functions
of the battery charge. As examples of continuous and discrete sources, we consider Gaussian and binary sources respectively. For the Gaussian case, we obtain a closed-form expression for the mismatch factor in terms of the $Lambert W$ function, and show that an increasing transmission power policy results in a decreasing mismatch factor policy and vice versa. Finally, we numerically compare the performance of the adaptive mismatch factor scheme to the case of a constant mismatch factor. 
\end{abstract}


\section{Introduction}
\IEEEPARstart{W}{ireless} sensor networks (WSN) provide a tool to gather and disseminate information and thus play an important role in unsupervised control systems. In many cases,
sensor devices are deployed in large numbers to cover a large area. As a result, regular maintenance and battery replacement for each individual sensor is impractical, if not impossible. Thus, to develop truly autonomous sensor networks that do not require regular maintenance, it is essential to supply a sustainable energy source to each node. 

Energy harvesting (EH), i.e., supplying energy by harnessing ambient energy resources such as solar, wind and thermal energy, is a promising state-of-the-art solution that can significantly improve sensor lifetime. In a typical EH node, the energy required for various sensor tasks is incrementally harvested from the environment and stored during the course of operation. However, due to the stochastic nature of renewable energy resources, sensor energy consumption must be managed adaptively to achieve good performance.
\subsection{Contributions}
In this paper, we focus on the design of data transmission policies in EH sensor devices. Specifically, we consider a scenario where a single node continuously senses data from a source and wishes to transmit this data over a point-to-point channel. The harvested energy is stored in a battery that may leak energy at a rate which depends on the available battery charge. The communication is carried by a joint source-channel coding (JSCC) scheme, where in each communication block a source sequence of fixed length is mapped into a channel codeword whose length depends only on the available battery charge. In other words, the mismatch factor between the channel symbols and source symbols is adaptive. 

We summarize the major contributions of this paper as follows:
\begin{itemize}
\item We formulate and model continuous-time lossy joint source-channel coding in an energy harvesting communication system, where the transmission power and bandwidth mismatch factor, i.e., the length of channel codewords per source symbol, are dynamically adapted to the available battery charge.
\item We establish a lower bound on the average distortion and show that in the case of infinite battery capacity and no leakage it is asymptotically achieved with a constant transmission power and a constant mismatch factor. 
\item Using a calculus of variations technique, we find achievable locally optimal transmission power and mismatch factor policies that minimize the average distortion at the receiver.
\item We develop an interesting structural result on the instantaneous distortion. Namely, as long as the battery is not depleted, locally optimal transmission power and mismatch factor policies will adaptively adjust with the battery charge in such a way that the instantaneous distortion is maintained to a constant level. 
\item For a moderate-size battery, we numerically show that our proposed scheme with an adaptively varying mismatch factor achieves a smaller average distortion compared to a scheme with a constant mismatch factor.
\item For Gaussian and binary sources, we numerically find locally optimal power policies and mismatch factor policies, both as functions of battery charge. With different leakage rates, i.e., zero leakage rate as well as arbitrary non-zero leakage rates, we observe that a good transmission power policy and mismatch factor policy increases and decreases respectively, as the battery charge increases.
\end{itemize}
\subsection{Related Work}\label{sec:2}
Among prior works that consider lossy communication with EH transmitters/receivers are \cite{zhao2013optimal,castiglione2012energy,castiglione2014energy,tapparello2012dynamic,orhan2014source, limmanee2013distortion,nayyar2013optimal}. In \cite{zhao2013optimal}, the mean squared distortion of the estimated source symbols at the receiver is minimized, where both online and offline EH scenarios are considered and the mismatch factor is always one. In \cite{castiglione2012energy}, the problem of energy allocation for data acquisition and transmission in WSNs is studied, and the case of a single sensor as well as the case of multiple sensors are both considered. A similar problem is studied in \cite{castiglione2014energy} for the case of a single source with finite battery and data buffer. Another interesting work is \cite{tapparello2012dynamic}, where a perturbation-based Lyapunov technique is used to obtain an online energy management scheme for source-channel coding of correlated sources. In \cite{orhan2014source}, a Gaussian source is transmitted over a flat fading channel and the offline minimization of the total distortion over a finite-time horizon is considered. Therein, the optimal distortion and transmission power are found subject to a delay constraint for reconstruction of the source symbols at the receiver. In \cite{limmanee2013distortion}, the problem of uncoded transmission over a fading channel is investigated, and an optimal energy allocation scheme to minimize the total distortion is established.   

Communication systems with EH transmitters/receivers have also been studied extensively in the context of lossless transmission. For instance, the offline minimization of the transmission completion time for a single source is considered in \cite{Yang2012a}, where the battery capacity is infinite. This problem has been also extended to the cases of finite battery \cite{tutuncuoglu2012optimum}, multiple access channels \cite{yang2012optimal}, broadcast channels \cite{yang2012broadcasting,ozel2012optimal,antepli2011optimal} and fading channels \cite{ozel2011transmission}. Also in \cite{DelayOptimal}, the problem of the trade-off between the average queueing delay of packets and power consumption of a  power grid has been studied using a 2D Markov chain model.

In another line of work, throughput/sum-throughput maximization is considered for point-to-point channels \cite{mitran2012optimal}, interference channels \cite{tutuncuoglu2012sum} and relay channels \cite{huang2013throughput,gunduz2011two}. A more realistic battery model with energy leakage is also considered in this context \cite{devillers2012general,tutuncuoglu2012communicating, tutuncuogluoptimum}. Specifically in \cite{tutuncuogluoptimum}, offline maximization of throughput in a single-user channel as well as a broadcast channel has been studied under battery imperfection constraints. In particular, the model assumes that only a fraction of harvested energy is stored in the battery due to charging/discharging inefficiency, whereas a completely different type of battery imperfection, i.e., battery leakage over time as in \cite{devillers2012general}, is considered in our paper.

Another similar work in the framework of lossless communication is \cite{khuzani2013online}, where authors have studied online maximization of long-term average sum-throughput in a multiple access channel using the same calculus of variations based technique. Despite similarity in the approach taken in both works, the nature of the problems are quite different. In the current work, we are concerned with lossy source-channel transmission in EH communication systems. The lossy nature of the problem opens a whole new set of questions that were not addressed in \cite{khuzani2013online}, e.g., how can the trade-off between increasing and decreasing transmission power be managed so that an optimal average distortion is achieved at the receiver? Does an adaptive source-channel coding scheme significantly improve the distortion? If so, how exactly should the length of channel codes per source symbols be changed in terms of the battery state? Our focus here is to minimize long-term average distortion in a point-to-point channel that exploits an adaptive joint source-channel coding scheme. Different from \cite{khuzani2013online}, we consider a more general battery model with a non-constant leakage rate and have numerically investigated the impact of different battery leakages on the compression-transmission performance. Moreover, in addition to the transmission power, the bandwidth mismatch factor between source symbols and channel codewords is also adapted to the battery charge.

The rest of this paper is organized as follows. In Section \ref{sec:3}, we study the communication model as well as the EH model. In Section \ref{sec:4}, we present the problem formulation. A lower bound on the average distortion as well as the achievable scheme of the communication resources are studied in Section \ref{sec:5}. Some structural results for the case of Gaussian source are established in Section \ref{sec:6}. In Section \ref{sec:8}, we provide numerical results and simulations. Finally, Section \ref{sec:9} concludes this paper.
\section{Preliminaries}\label{sec:3}
\subsection{Communication Model}
We consider the lossy source-channel transmission of a stationary memoryless source with general alphabet (continuous or discrete) over a point-to-point channel. We first assume that the communication is carried over $K$ consecutive blocks of joint source-channel coding (JSCC), where the number of source symbols in each block is fixed, whereas the length of the channel codewords varies from one block to another based on the available charge in the battery. Specifically, during the $i^{\rm th}$ block, $i=1, ..., K$, a sequence of $m$ independent and identically distributed (i.i.d.) realizations of the source $s_i^m = (s_i[1], s_i[2], ..., s_i[m])$, are mapped to a channel codeword of length $n_i$, i.e., $x_i^{n_i} = (x_i[1], x_i[2], ..., x_i[n_i])$. The input $x_i^{n_i}$ induces a distribution on the channel output $y_i^{n_i} = (y_i[1], y_i[2], ..., y_i[n_i])$, according to the law $\mathbb{P}_Y \left( y_i^{n_i} \right)= \prod_{k=1}^{n_i} \mathbb{P}_{Y \vert X}( y_i[k] \big\vert x_i[k]) \mathbb{P}_X \left( x_i[k] \right)$, where $\mathbb{P}_{Y \vert X}\left( y \vert x \right)$ is the conditional distribution of the stationary and memoryless channel. At the receiver, after each observation $y_i^{n_i}$, $i=1, ..., K$, an estimate $\hat{s}_i^m(y_i^{n_i}) = (\hat{s}_i[1], \hat{s}_i[2], ..., \hat{s}_i[m])$ of $s_i^m$ is made.

We assume that $\tau$ is the time duration needed for one symbol to be generated/transmitted. Associated with $\tau$, $\Delta t_s=m \times \tau$ is the time duration to generate $m$ source symbols, and $\Delta t_{c_i}=n_i \times \tau$ is the time duration to transmit a channel codeword of length $n_i$ during the $i^{\rm th}$ block (see Fig. \ref{fig:1}). 
\begin{definition}
We define the bandwidth mismatch factor between the $i^{\rm th}$ channel codeword block, with duration $\Delta t_{c_i}$, and the source symbol block, with duration $\Delta t_s$, as $‎\kappa‎(t_{c_i}):=\Delta t_{c_i}/\Delta t_s=n_i/m$ for $i=1, ..., K$, where $t_{c_i}:=\sum_{j=1}^i \Delta t_{c_j}$ is the channel output time (at the end of the $i^{\rm th}$ block) that we take as the reference time throughout the paper. Likewise, $\lbrace p(t_{c_i}) \rbrace_{i=1}^K$ are the average power constraints on the codewords, i.e.,
\begin{equation}\label{def}
\frac{1}{n_i} \sum_{j=1}^{n_i} \big\vert x_i[j] \big\vert^2 \leq p(t_{c_i}), \quad  i=1, ..., K.
\end{equation} 
\end{definition}
As demonstrated in Fig. \ref{fig:1}, the mismatch factor $\kappa(t_{c_i})$ is fixed throughout each block, nevertheless it can change adaptively from one block to another. Similarly, the transmission power $p(t_{c_i})$ can change from block to block. 
\begin{figure*}[t!]
\centering 
\def\svgwidth{300pt} 
\hspace{1cm}\input{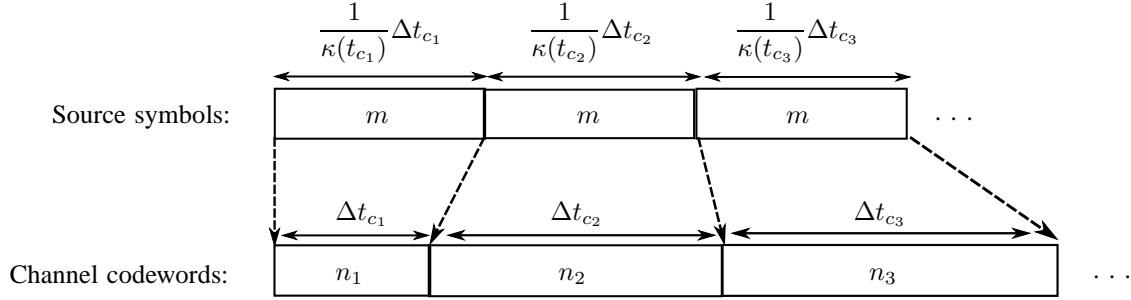} 
\caption{\small Consecutive blocks of JSCC, where $m$ and $n_i$, $i=1, ..., K$, are the number of source symbols and length of channel codewords in each block, respectively.} 
\label{fig:1}
\end{figure*} 
The rate-distortion function (i.e., lossy source coding rate) in the $i^{\rm th}$ block, for a source $S$ is given by
\begin{equation}\label{eq:3}
R_s(D_i)= \inf_{} \; I(\hat{S};S), \quad i=1, ..., K,
\end{equation}
where the infimum is taken over all conditional distributions $\mathbb{P}_{\hat{S}\vert S}$ such that $\mathbb{E}[d(\hat{S},S)] \leq D_i$ in the $i^{\rm th}$ block, and $I(\hat{S};S)$ is the mutual information between the estimated symbols $\hat{S}$ and the source symbols $S$. We assume that for a given source $S$ and distortion measure $d(\hat{s},s)$ the rate-distortion function $R_s(D)$ has two thresholds $D_{\rm max}$ and $R_s^{\rm th}$ (where $D_{\rm max}$ is always finite whereas $R_s^{\rm th}$ could be finite or infinite), such that:
\begin{description}
\item[{[S1]}]  $R_s(D)$ is zero for $D \geq D_{\rm max}$,
\item[{[S2]}] $R_s(D)$ is strictly decreasing, convex and twice continuously differentiable for $0 < D < D_{\rm max}$, i.e., $R_s^{\prime}(D) < 0$ and $R_s^{\prime\prime}(D) > 0$, and continuous at $D=D_{\rm max}$,
\item[{[S3]}] $\lim_{D \downarrow 0}R_s^{\prime}(D)=-\infty$ , and $R_s^{\prime}(D)$ is finite everywhere else,
\item[{[S4]}] $\lim_{D \downarrow 0}R_s(D)=R_s^{\rm th}$, and $R_s(D)$ is right-continuous at $D=0$ if $R_s^{\rm th}$ is finite.
\end{description}
These conditions abstract the form of a rate-distortion function. For example, the Gaussian source $\mathcal{N}(0,\sigma^2)$, with squared-error distortion measure \textbf{$d(\hat{s},s)=\vert s - \hat{s} \vert^2 $} has the rate-distortion function given by \cite{cover2012elements},
\begin{equation}\label{eq:3:1}
 R_s(D) = \left\{ 
  \begin{array}{l l}
    \dfrac{1}{2} \log \dfrac{\sigma^2}{D}\; & \quad 0 < D < \sigma^2 \\
    0\; & \quad D \geq \sigma^2,
  \end{array} \right.
\end{equation}
and satisfies [S1]-[S4] with $D_{\rm max}=\sigma^2$ and $R_s^{\rm th}=\infty$. Another example is a Bernoulli$(\mathsf{p})$ binary source with Hamming distortion $d(\hat{s},s)=0$ for $s=\hat{s}$ and $d(\hat{s},s)=1$ otherwise. The rate-distortion function is then given by \cite{cover2012elements},
\begin{equation}\label{eq:3:2}
R_s(D) = \left\{ 
  \begin{array}{l l}
    H(\mathsf{p})-H(D)\; & \quad 0 \leq D < \min \lbrace \mathsf{p}, 1-\mathsf{p} \rbrace \\
    0\; & \quad D \geq \min \lbrace \mathsf{p}, 1-\mathsf{p} \rbrace,
  \end{array} \right.
\end{equation}
where $H(\mathsf{p})$ is the entropy of the source and
\begin{equation}\label{eq:3:3}
H(D) := -D \log_2 D - (1-D) \log_2 (1-D).
\end{equation}
For the binary source, $D_{\rm max}=\min \lbrace \mathsf{p}, 1-\mathsf{p} \rbrace$ and $R_s^{\rm th}=H(\mathsf{p})$. 

In the $i^{\rm th}$ block, the transmit power is limited to $p(t_{c_i})$, and thus the channel coding rate $R_c(p(t_{c_i}))$ is given by
\begin{equation}
R_c(p(t_{c_i})) = \sup I(X;Y), \quad i=1, ..., K,
\end{equation}
where the supremum is taken over all channel input distributions $\mathbb{P}_X$ that satisfy the power constraint $p(t_{c_i})$ in \eqref{def} for the $i^{\rm th}$ block. This model allows for fast fading, where the block length is sufficiently large such that the effect of fading can be averaged. We assume that the channel coding rate $R_c(p)$ has the following properties:
\begin{description}
\item[{[C1]}] $R_c(p)$ is strictly positive for $p>0$, zero and right-continuous at $p=0$,
\item[{[C2]}] $R_c(p)$ is strictly increasing, concave and twice continuously differentiable  for $p>0$, i.e., $R_c^{\prime}(p) > 0$ and $R_c^{\prime\prime}(p) < 0$.
\end{description}
Similarly, these conditions abstract the form of a channel coding rate function, and an example is the Shannon rate function $R_c(p)=\dfrac{1}{2}\log_2(1+p/N)$, which will be assumed later in the paper.

We assume that the block lengths $m$ and $n_i, \; i=1, ..., K$, are sufficiently large that the rate-distortion function and the channel coding rate have operational significance. Since the mismatch factor is fixed throughout each JSCC block, source-channel separation holds in each block. Therefore, based on a separate source and channel coding scheme, the relation between source coding rate $R_s(D_i)$ and the channel coding rate $R_c(p(t_{c_i}))$ is given by 
\begin{equation}\label{eq:4}
R_s(D_i)= \kappa(t_{c_i}) R_c(p(t_{c_i})), \quad i=1, ..., K.
\end{equation} 
One should note that by the condition [S2] and continuity of $R_s(D_i)$ at $D_i=D_{\rm max}$, the inverse function $D_i(R_s)$ of the rate-distortion always exists for $0 \leq \kappa(t_{c_i}) R_c(p(t_{c_i})) < R_s^{\rm th}$ (though in most cases there is no closed-form), and we can compute the distortion $D_i(R_s)$ in the $i^{\rm th}$ block in terms of the transmission power $p(t_{c_i})$ and the mismatch factor $\kappa(t_{c_i})$ using \eqref{eq:4}, i.e., $D_i(R_s) := D(p(t_{c_i}),\kappa(t_{c_i}))$. Moreover, if $R_s^{\rm th}$ is finite then the distortion is $D(p(t_{c_i}),\kappa(t_{c_i}))=0$ for pairs $p(t_{c_i}),\kappa(t_{c_i})$ such that $\kappa(t_{c_i}) R_c(p(t_{c_i})) \geq R_s^{\rm th}$.

From the definition of the mismatch factor and Fig. \ref{fig:1}, we directly obtain
\begin{equation}\label{eq:1}
K\times \Delta t_s = \sum_{i=1}^K \frac{1}{\kappa(t_{c_i})} \Delta t_{c_i}.
\end{equation}
Similarly, the total average distortion per source symbol over $K$ blocks can be written as
\begin{align}\label{eq:5-1}
D_{\rm avg}:=& \hspace{2mm} \frac{1}{K\Delta t_s} \sum_{i=1}^K D(p( t_{c_i}),\kappa( t_{c_i})) \Delta t_s \\
\label{eq:5}=& \hspace{2mm} \frac{1}{K\Delta t_s} \sum_{i=1}^K D(p( t_{c_i}),\kappa( t_{c_i})) \frac{1}{\kappa( t_{c_i})} \Delta t_{c_i}.
\end{align}   

\subsection{Continuous-Time Model}
As a practical assumption, we suppose that both the source and channel block times, $\Delta t_s$ and $\Delta t_{c_i}$ respectively, are small compared to the battery dynamics. Hence, we develop a continuous-time model in the asymptotic regime based on \eqref{eq:1} and \eqref{eq:5}, and let $K$ grow so that
\begin{align}
\nonumber T_s :=& \; K \Delta t_s  \\
\nonumber  T_c :=&  \; \sum_{i=1}^K \Delta t_{c_i},
\end{align}
where the terms $T_s$ and $T_c$ have fixed values as $\Delta t_s \to 0$, $\Delta t_{c_i} \to 0$ and $K \to \infty$. Eq. \eqref{eq:1} has the form of a Riemann sum in the limit of $\Delta t_s \to 0$, $\Delta t_{c_i} \to 0$ and $K \to \infty$,  and thus the continuous-time limit takes the following form
\begin{align}\label{eq:6}
T_s = \int_0^{T_c} \frac{1}{\kappa(t_c)} dt_c.
\end{align}
We define 
\begin{align}
\label{eq:6:1}\rho(T_c) := \; \frac{T_s}{T_c} = \frac{1}{T_c} \int_0^{T_c} \frac{1}{\kappa(t_c)} dt_c.
\end{align}
In this paper, we are interested in infinite-time horizon communication and long-term average distortion. We thus let the transmission time become asymptotically large, i.e., $T_c \to \infty$. If $\lim_{T_c \to \infty} \rho(T_c) > 1$ or $\lim_{T_c \to \infty} \rho(T_c) < 1$, the backlog in either the source symbols queue or the transmission of channel codewords tends to infinity. Thus, for a stable source-channel communication system we assume $\lim_{T_c \to \infty} \rho(T_c) = \lim_{T_c \to \infty} T_s/T_c = 1$. Therefore, from \eqref{eq:6:1} for an asymptotically stable joint source-channel communication system we require that,
\begin{align}\label{eq:7}
 \lim_{T_c \to \infty} \frac{1}{T_c} \int_0^{T_c} \frac{1}{\kappa(t_c)} dt_c = 1.
\end{align} 
Hence, we rewrite the continuous-time limit of \eqref{eq:5} as 
\begin{align}
\nonumber D_{\rm avg} &= \lim_{T_c \to \infty} \frac{1}{T_s} \int_0^{T_c} D(p(t_c),\kappa(t_c)) \frac{1}{\kappa(t_c)} dt_c \\
\label{eq:8} &= \lim_{T_c \to \infty} \frac{1}{T_c} \int_0^{T_c} D(p(t_c),\kappa(t_c)) \frac{1}{\kappa(t_c)} dt_c,
\end{align}
where \eqref{eq:8} follows from the fact that $\lim_{T_c \to \infty} T_s/T_c = 1$.
\subsection{Storage Model}
We assume that the energy arrival times and their amounts are not known at the transmitter a priori. Therefore, the instantaneous energy of the battery is a stochastic process that can be characterized in terms of the energy arrival process and battery depletion rate. Let $\{E_i\}_{i=1}^{\infty}$, denote the size of the energy packets arriving to the battery at time instants $\{T_i\}_{i=1}^{\infty}$, where $E_i>0$ and $T_1 < T_2 < ...$ . We assume that energy packets are i.i.d. with the tail distribution function denoted by $B(z)=\mathbb{P} [E > z]$, and the corresponding arrival times are a homogeneous Poisson Point Process (PPP). Specifically, the inter-arrival times are i.i.d. and exponentially distributed with parameter $\delta$, i.e., $\Delta T_n := T_{n+1} - T_n \sim {\rm Exp}(\delta) $. One should note that the PPP assumption can subsume cases of EH systems that are modeled as bursty packet arrivals at discrete times such as regenerative shock absorbers that use the piezoelectric effect to transform random mechanical shocks into electrical energy. Due to the unpredictability of their occurrence, a PPP is a good model to describe the behavior of such EH systems. 

The total harvested energy at the transmitter up to the time $t$, $\{A(t): t \geq 0 \}$ is thus a compound Poisson process given by
\begin{equation}\label{eq:9}
A(t) :=\sum_{n\in \mathbb{N}} E_n \mathds{1}_{\lbrace T_n <t \rbrace} .
\end{equation} 
We assume that the capacity of the battery $L$ is finite, that is, $L < \infty$. Furthermore, we assume that the battery is imperfect in the sense that when it is charged, it leaks energy over time at a rate which depends only on the current battery charge $Z(t)$, and is denoted by $\ell(t)$ at time $t$. Also, it is clear that there is no leakage when the battery is depleted. 

The instantaneous battery charge at time $t$, $\{ Z(t): t \geq 0 \}$ is therefore a stochastic process described by
\begin{equation}\label{eq:10}
Z(t)= z_0 + A(t) - \int_{0^+}^t \Big( p \left(s\right)+ \ell (s) \Big) ds- R(t),
\end{equation}    
where $z_0 = Z(t)\vert_{t=0}$ is the initial battery charge and $R(t)$ is a non-decreasing, non-negative and continuous-time \textit{reflection} process with $R(t)\vert_{t=0} = 0$, that only increases over the set $\{t:Z(t)=L\}$ \cite{piera2008boundary}. The reflection process accounts for the excess energy arrivals that overflow the battery capacity and ensures that even for large energy packet arrivals the storage process is bounded by its capacity limit at all times, i.e., $Z(t) \in [0 , L]$. Furthermore, $p(t)+ \ell (t)$ is the instantaneous battery depletion rate. Since the transmission power $p(t)$ is adapted to the battery charge $Z(t)$, this depends on time only through $Z(t)$, i.e., $p(t) + \ell (t)=p \left( Z(t) \right)+ \ell \left( Z(t) \right)$. Likewise, the instantaneous mismatch factor $\kappa(t)$ is modulated by the battery charge $Z(t)$, where  $\kappa(t)= \kappa(Z(t))$. More specifically, since $p(t)$, $\ell(t)$ and $\kappa(t)$ all depend on $t$ only through the battery charge $Z(t)$, with slight abuse of notation we denote by $p(z)$, $\ell(z)$ and $\kappa(z)$ the explicit dependence of these on $z$. In the rest of this paper, we refer to $p(z)$ and $\kappa(z)$ as the power policy and mismatch factor policy. The storage process $Z(t)$ can then be viewed as a continuous-time Markov process, where the state space of the process is the finite interval $[0,L]$. We impose the following conditions on the feasible set of power policies and leakage rates
\begin{itemize}
\item $\forall z \in (0,L] ,\; p(z) > 0 $, and $p(z)\big|_{z=0}=0$,
\item $\sup_{0 < z \leq L}  p(z) < \infty $, 
\item $\sup_{0 < z \leq L}  \ell(z) < \infty $, and $\ell(z)\big|_{z=0}=0$.
\end{itemize}   
The first condition is to avoid a reserve of energy in the battery that can never be consumed by transmission and thus effectively reduces the usable energy stored in the battery. The second condition reflects the fact that instant depletion of an amount of energy $\Delta E>0$ is not allowed. Thus, an optimal power policy $p(z)$ must satisfy both of these constraints. With these conditions on $p(z)$ and $\ell(z)$, $Z(t)$ becomes irreducible in the sense that there is only one single communicating class in the state space. 

We define $\tilde{\pi}_{T_c}(z)$ as the empirical distribution function of the storage process with respect to the reference time, i.e.,
\begin{equation}\label{eq:14}
\tilde{\pi}_{T_c}(z):= \frac{1}{T_c}\int_0^{T_c} \mathds{1}_{\lbrace Z(t_c) \leq z\rbrace} dt_c.
\end{equation}
By the strong law of large numbers, as $T_c \to \infty$, $\tilde{\pi}_{T_c}(z)$ converges to the stationary probability measure of the storage process denoted by $\pi(z)$, almost surely for every value of $z$.

The following theorem identifies stationary distribution of the storage process $Z(t)$ \cite{asmussen2003applied}, where we recall that the inter-arrival times are i.i.d. and exponentially distributed with parameter $\delta$, i.e., $T_{n+1} - T_n \sim {\rm Exp}(\delta) $.
\begin{theorem}
For $L<\infty$, the storage process $Z(t)$ is positive recurrent and there exists a unique stationary probability measure $\pi(z)=\mathbb{P}[Z(t) \leq z]$, which may have an atom $\pi_0 := \pi(z)\vert_{z=0} \geq 0$, and is absolutely continuous on $(0,L]$ such that
\begin{equation}\label{eq:12}
\pi(z) =  \pi_0 + \int_{0^+}^z f(u) du ,
\end{equation}
where $f(z)$ is the absolutely continuous part of the probability measure. Moreover,
\begin{equation}\label{eq:13}
f(z) \left( p(z)+ \ell (z) \right) = \delta \pi_0 B(z) + \delta \int_{0^+}^z B(z-u) f(u) du.
\end{equation}
\end{theorem}

\begin{remark}
From \eqref{eq:14}, $\pi_0$ is the fraction of time that the battery remains discharged. 
\end{remark}
\begin{remark}
Equation \eqref{eq:13} is the equilibrium condition between the rate of down-crossing $f(z) ( p(z) + \ell(z) )$ and the rate of up-crossing $\delta \pi_0 B(z) + \delta \int_{0^+}^z B(z-u) f(u) du$ at level $z$.
\end{remark}
We now assume that the packets of energy are exponentially distributed with parameter $\lambda$, i.e., $B(z) = \exp(-\lambda z)$. Therefore, \eqref{eq:13} reduces to
\begin{equation}\label{eq:16}
f(z) \left( p(z)+ \ell (z) \right) = \delta e^{-\lambda z} \left( \pi_0 + \int_{0^+}^z e^{\lambda u} f(u) du \right).
\end{equation}

\section{Problem Formulation}\label{sec:4}
Since, $\pi(z)$ is a probability measure we have
\begin{align}\label{eq:17}
\int_0^L \pi(dz) &= \pi_0 + \int_{0^+}^L f(z) dz = 1.
\end{align}
By combining \eqref{eq:7} and \eqref{eq:12}, we also obtain the following constraint on the mismatch factor
\begin{align}
 ‎\lim_{T_c \to \infty} ‎\frac{1}{T_c}‎ ‎\int_0^{T_c} ‎\frac{1}{\kappa(Z(t_c))} dt_c  &=‎ \mathbb{E}_{\pi}\left[ \frac{1}{\kappa(z)} \right] \quad \quad a.s. \\
 &= \frac{\pi_0}{\kappa(0)} + \int_{0^+}^L \frac{f(z)}{\kappa(z)}  dz \\
 &=1, \label{eq:20} 
\end{align}
where $\kappa(0):=\kappa(z)\vert_{z=0}$ is the number of channel uses per source symbol when the battery is exhausted. Similarly, we apply the ergodicity argument to the distortion function in \eqref{eq:8}. We first rewrite $D_{\rm avg}$ from \eqref{eq:8} as
\begin{equation}\label{eq:21}
D_{\rm avg}=\lim_{T_c \to \infty} \frac{1}{T_c} \int_0^{T_c} \hspace{-3mm}D \Big( p \left( Z \left( t_c \right) \right),\kappa \left( Z \left( t_c \right) \right)  \Big) \frac{1}{\kappa \left( Z \left( t_c \right) \right) } dt_c.
\end{equation} 
Due to the ergodicity of the storage process $Z(t_c)$, we obtain
\begin{align}\label{eq:22}
D_{\rm avg} =& \; \int_0^L D\left(p(z),\kappa(z)\right)\frac{1}{\kappa(z)} \pi(dz) \quad \quad a.s. \\
\label{eq:23} =& \; \pi_0 D^{\dagger}(p(0),\kappa(0)) + \int_{0^+}^L D^{\dagger}\left(p(z),\kappa(z)\right) f(z) dz\\ 
=& \; \mathbb{E}_{\pi}\left[D^{\dagger}\left(p(Z),\frac{1}{\kappa(Z)}\right)\right],
\end{align}
where
\begin{equation}\label{eq:24}
D^{\dagger} \left( p, \frac{1}{\kappa} \right) := D (p,\kappa) \frac{1}{\kappa}.
\end{equation}
We now wish to find an optimal transmission power policy $p(z)$ and mismatch factor policy $\kappa(z)$ such that the distortion at the receiver is minimized. More specifically, we want to minimize \eqref{eq:23} subject to the constraints in \eqref{eq:16}, \eqref{eq:17} and \eqref{eq:20}, i.e.,
\begin{align}\label{eq:25}
&\hspace{-8mm}\inf_{f(z),\pi_0,\kappa(z),\kappa(0)}  \hspace{-1mm}   \pi_0 D^{\dagger}(p(0),\kappa(0)) + \int_{0^+}^L D^{\dagger}\left(p(z),\kappa(z)\right) f(z) dz\\
\label{eq:25:1}
\mbox{s.t.} \hspace{1mm}& f(z)\left( p(z) + \ell(z) \right) =\delta e^{-\lambda z}\Big(\pi_0 + \int_{0^+}^L e^{\lambda u}f(u)du\Big) \\ 
\label{eq:25:2}
&\pi_0+\int_{0^+}^L f(z) dz =1\\
\label{eq:25:3}
& \frac{\pi_0}{\kappa(0)} + \int_{0^+}^L \frac{f(z)}{\kappa(z)} dz =1\\
\label{eq:25:4}
&f(z)\geq0,\quad \pi_0\geq0,\quad \kappa(z)>0,\quad  \kappa(0) > 0,
\end{align}
where the strict inequalities $\kappa(z)>0$ and $\kappa(0)>0$ are implied by \eqref{eq:25:3}. We note that for a given $f(z)$, $p(z)$ can be calculated directly using \eqref{eq:25:1}. 

We note that in the formulation of the optimization problem \eqref{eq:25}, the mismatch factor and the transmission
power can be viewed as two limited resources that must be allocated efficiently. Specifically, due to \eqref{eq:25:3}, the inverse mismatch factor has an average of unity. Consequently, using
a large mismatch factor, which is desirable as this reduces the instantaneous distortion, must be balanced with using a small mismatch factor to maintain the average constraint. Similarly, the transmitter cannot consume more energy than what is stored in the battery, which itself is replenished at a maximum rate of $\delta/\lambda$ (ignoring the energy that is lost due to overflow). Thus, periods with a large transmission power must be accompanied by periods with a small transmission power.

\section{Main Results}\label{sec:5}
\subsection{Distortion Lower Bound}
In this part, we derive a distortion lower bound using Jensen's inequality and convexity of the distortion function.
\subsubsection{Finite Battery Capacity}
We first note that for any power policy $p(z)$, when the battery capacity is finite, the following upper bound has been computed in \cite{khuzani2013online}:
\begin{align}
\mathbb{E}_{\pi} \left[p(Z)\right] 
\leq \frac{\delta}{\lambda} \left( 1 - e^{-\lambda L} \right).
\label{eq:26:1}
\end{align}
One should note that even with a non-zero leakage rate, the upper bound in \eqref{eq:26:1} still holds, however it is a potentially looser bound. 

We now find a lower bound on the objective function $\mathbb{E}_{\pi} [ D^{\dagger} (p(Z), 1/\kappa(Z))]$, based on the upper bound in \eqref{eq:26:1} and the following convexity lemma.
\begin{lemma}\label{lemma:1}
The function $D^{\dagger} (p,q)$ defined in \eqref{eq:24} is jointly convex over the pair $p$ and $q$, where $q:=1/\kappa$.
\label{th:1}
\end{lemma}
\begin{proof}
See Appendix A. 
\end{proof}

Based on Lemma \ref{lemma:1} as well as Jensen's inequality, we establish a distortion lower bound as below
\begin{align}\label{eq:28}
\mathbb{E}_{\pi}\left[D^{\dagger}\left(p(Z),\frac{1}{\kappa(Z)}\right)\right] &\geq D^{\dagger}\left(\mathbb{E}_{\pi}\left[p(Z)\right],\mathbb{E}_{\pi}\left[\frac{1}{\kappa(Z)}\right]\right) \\ 
\label{eq:28:2}&\geq  D^{\dagger}\left(\frac{\delta}{\lambda} \left( 1 - e^{-\lambda L} \right), 1 \right)\\
\nonumber &\hspace{-1mm}:= D_{\rm LB},
\end{align}
where \eqref{eq:28:2} follows from \eqref{eq:26:1} and the fact that the function $D^{\dagger}(p,1/\kappa)$ is a non-increasing function of $p$, and further recalling that $‎\mathbb{E}_{\pi}[1/\kappa(Z)]=1‎$ from \eqref{eq:20}. The lower bound in \eqref{eq:28:2} holds for any transmission power $p(z)$ and mismatch factor $\kappa(z)$ that satisfies \eqref{eq:25:3} and \eqref{eq:25:4}.
\subsubsection{Infinite Battery Capacity}
When the capacity of the battery is infinite, i.e., $L \to \infty$, and therefore the chance of battery overflow is zero, for any ergodic power policy $p(z)$ the lower bound in \eqref{eq:28:2} simplifies to
\begin{align}\label{eq:95-1} 
\mathbb{E}_{\pi}\left[D^{\dagger}\left(p(Z),\frac{1}{\kappa(Z)}\right)\right] &\geq D^{\dagger}\left( \frac{\delta}{\lambda} ,1 \right).
\end{align}
One should note that if the leakage rate is zero, every ergodic power policy $p(z)$ is such that $\mathbb{E}_{\pi}\left[ p(Z) \right]=\delta/\lambda$. This suggests that in the infinite battery capacity case with no leakage, separate source and channel coding with a constant mismatch factor $\kappa(z)=1$ is optimal. More precisely, consider the following choice of power policy
\begin{equation}\label{eq:95}
p^{\star}(z) = \frac{\delta}{\lambda} + \epsilon, \quad z>0,
\end{equation} 
with $p^{\star}(0)=0$, and $\epsilon>0$ is a small positive number that ensures the storage process is positive recurrent, and choose $\kappa^{\star}(z)=1$ for $z\geq0$. Then,
\begin{equation}\label{eq:97}
\mathbb{E}_{\pi}\left[ p^{\star}(Z) \right] = \left( \frac{\delta}{\lambda}+\epsilon \right) \left( 1-\pi_0 \right), 
\end{equation}
and since we must have $\mathbb{E}_{\pi}[p^{\star}(z)]=\delta/\lambda$, then
\begin{equation}\label{eq:98}
1- \pi_0 = \frac{\delta / \lambda}{\delta / \lambda +\epsilon}.
\end{equation}
In particular, we then compute the total average distortion of this scheme as
\begin{equation}\label{eq:98:1}
\mathbb{E}_{\pi}\hspace{-1mm}\left[D^{\dagger}\left(p^{\star}(Z),\frac{1}{\kappa^{\star}(Z)}\right)\right]\hspace{-1mm}=\hspace{-1mm}(1-\pi_0) D^{\dagger}\left( \frac{\delta}{\lambda}+\epsilon ,1 \right) + \pi_0 D_{\rm max}.
\end{equation}
As $\epsilon \to 0$ the atom $\pi_0$ tends to zero from \eqref{eq:98} and the lower bound in \eqref{eq:95-1} is thus asymptotically achieved.
\subsection{Achievable Resource Allocation Scheme}
The two functions $p(z)$ and $\kappa(z)$ are the policies that we wish to design in this subsection to manage the limited resources of the system in such a way that the average distortion in \eqref{eq:25} is minimized. To this end, we use a calculus of variations technique which provides necessary conditions for a local and therefore global optimal solution to our optimization problem. 
\begin{figure*}[!t]
\vspace{-0.3cm}
\normalsize
\setcounter{equation}{62}
\begin{align}  
\nonumber &\mathbb{E}_{\pi^{\epsilon}}‎\left[D^{\dagger}\left(p^{\epsilon}(Z),\frac{1}{\kappa(Z)} \right)\right] \\ 
\label{eq:45} &=  \pi_0 D^{\dagger}(p(0), \kappa(0)) + \int_{0^+}^L   \left( D(p(z),\kappa(z)) 
+ \epsilon \frac{\partial D(p(z),\kappa(z))}{\partial p(z)} \frac{d p^{\epsilon}(z)}{d \epsilon} \Big|_{\epsilon=0} \right) \times \left( \frac{f(z)+ \epsilon h(z)}{\kappa(z)} \right) dz \\
 &= \mathbb{E}_{\pi} ‎\left[D^{\dagger}\left(p(Z),\frac{1}{\kappa(Z)} \right)\right]
 + \epsilon \int_{0^+}^L \left( D(p(z),\kappa(z)) \frac{h(z)}{\kappa(z)} +  \frac{\partial D(p(z),\kappa(z))}{\partial p(z)} \frac{d p^{\epsilon}(z)}{d \epsilon} \Big|_{\epsilon=0} \frac{f(z)}{\kappa(z)} \right) dz, \label{eq:46}
\end{align}
\hrulefill
\vspace*{4pt}
\end{figure*}
\setcounter{equation}{42}

\subsubsection{Calculus of Variations}
We define $f^{\epsilon}(z)$ and $1/\kappa^{\epsilon}(z)$ as a perturbed density function and perturbation of the inverse mismatch factor, respectively, i.e.,
\begin{align}\label{eq:29}
f^{\epsilon}(z) &:=f(z)+\epsilon h(z), \\
\label{eq:29:1}\frac{1}{\kappa^{\epsilon}(z)} &:=\frac{1}{\kappa(z)}+\epsilon g(z),
\end{align}
where $h(z)$ and $g(z)$ are continuous and bounded perturbation functions on $(0,  L]$, with $h(0^+)=h(L)=0$ and $g(0^+)=g(L)=0$. For sufficiently small $\epsilon>0$, the perturbed density function $f^{\epsilon}(z)$ satisfies \eqref{eq:25:2} only if
\begin{align}\label{eq:30}
\pi_0 + \int_{0^+}^L f^{\epsilon}(z) dz &= \pi_0 + \int_{0^+}^L f(z) dz + \epsilon \int_{0^+}^L h(z) dz \\
&= 1,
\end{align}
which due to \eqref{eq:17} is true for all $\epsilon > 0$ iff
\begin{equation}\label{eq:31}
\int_{0^+}^L h(z)dz =0 .
\end{equation}
In addition, from \eqref{eq:25:3} we derive the two following conditions 
\begin{align}
\label{eq:33}  \frac{\pi_0}{\kappa(0)} + \int_{0^+}^L \frac{f(z)}{\kappa^{\epsilon}(z)} dz &=1,\\
\label{eq:32}  \frac{\pi_0}{\kappa(0)} + \int_{0^+}^L \frac{f^{\epsilon}(z)}{\kappa(z)} dz &=1.
\end{align}
We simplify \eqref{eq:33} using \eqref{eq:29:1} as follows
\begin{align}
\nonumber \frac{\pi_0}{\kappa(0)} +& \int_{0^+}^L \left( \frac{1}{\kappa(z)} + \epsilon g(z) \right) f(z) dz\\
 &= \frac{\pi_0}{\kappa(0)} + \int_{0^+}^L  \frac{f(z)}{\kappa(z)} dz + \epsilon \int_{0^+}^L  g(z)f(z) dz \label{eq:36} \\
\label{eq:36:1} &= 1.
\end{align}
Similarly, we simplify \eqref{eq:32} using \eqref{eq:29} as
\begin{align}
\nonumber \frac{\pi_0}{\kappa(0)} +& \int_{0^+}^L  \frac{f^{\epsilon}(z)}{\kappa(z)} dz\\
&= \frac{\pi_0}{\kappa(0)} + \int_{0^+}^L  \frac{f(z)}{\kappa(z)} dz + \epsilon \int_{0^+}^L  \frac{h(z)}{\kappa(z)} dz \label{eq:36:2} \\
\label{eq:36:3} &= 1.
\end{align} 
Thus, analogous to the constraint in \eqref{eq:31}, both \eqref{eq:36} and \eqref{eq:36:2} are true for sufficiently small $\epsilon$ iff the perturbation functions also satisfy the following constraints
\begin{align}
\label{eq:37:1}\int_{0^+}^L g(z)f(z)dz &= 0\\
\label{eq:37}\int_{0^+}^L \frac{h(z)}{\kappa(z)} dz &= 0.
\end{align}
\begin{figure*}[!t]
\vspace{-0.3cm}
\normalsize
\setcounter{equation}{72}
\begin{align}  
\nonumber &\mathbb{E}_{\pi} \left[D^{\dagger}\left(p(Z),\frac{1}{\kappa^{\epsilon}(Z)} \right) \right] \\ 
 &= \mathbb{E}_{\pi} \left[D^{\dagger}\left(p(Z),\frac{1}{\kappa(Z)} \right) \right]
 + \epsilon \int_{0^+}^L \left( D(p(z),\kappa(z)) g(z)f(z)- \frac{\partial D(p(z),\kappa(z))}{\partial \kappa(z)} \kappa(z) g(z)f(z) \right) dz.
\label{eq:57} 
\end{align}
\hrulefill
\vspace*{4pt}
\end{figure*}
\setcounter{equation}{55}
For sufficiently small $\epsilon$, two necessary conditions for a local and therefore a global optimal solution to the optimization problem in \eqref{eq:25}-\eqref{eq:25:4} are
\begin{align}\label{eq:38}
\hspace{-3mm}\mathbb{E}_{\pi^{\epsilon}}‎\left[D^{\dagger}\left(p^{\epsilon}(Z),\frac{1}{\kappa(Z)} \right)\right] &‎\geq \mathbb{E}_{\pi} \left[D^{\dagger}\left(p(Z),\frac{1}{\kappa(Z)} \right) \right],\\
\label{eq:39}
\hspace{-3mm}\mathbb{E}_{\pi} \left[D^{\dagger}\left(p(Z),\frac{1}{\kappa^{\epsilon}(Z)} \right) \right] ‎&\geq \mathbb{E}_{\pi} \left[D^{\dagger}\left(p(Z),\frac{1}{\kappa(Z)} \right) \right],
\end{align}
where $\mathbb{E}_{\pi^{\epsilon}}$ is the expectation with respect to the probability measure with the perturbed density function $f^{\epsilon}(z)$, (see \eqref{eq:29}). 

We now expand the l.h.s of \eqref{eq:38} as below 
\begin{align}
\nonumber &\mathbb{E}_{\pi^{\epsilon}}‎\left[D^{\dagger}\left(p^{\epsilon}(Z),\frac{1}{\kappa(Z)} \right)\right] \\ &\hspace{2mm}= \pi_0 D^{\dagger} \left( p(0), \kappa(0) \right) 
+ \int_{0^+}^L D^{\dagger}(p^{\epsilon}(z),\kappa(z)) f^{\epsilon}(z) dz.\label{eq:41}
\end{align}
We then use \eqref{eq:25:1} to compute $p^{\epsilon}(z)$ as follows 
\begin{align}
\nonumber p^{\epsilon}&(z) \\ 
 &\hspace{-1mm}= \delta e^{-\lambda z} \frac{\left( \pi_0 + \int_{0^+}^z e^{\lambda u} f^{\epsilon}(u) du \right) }{f^{\epsilon}(z)} -\ell (z) \label{eq:40} \\
&= \delta e^{-\lambda z} \frac{\left(\pi_0 + \int_{0^+}^z e^{\lambda u} f(u) du +\epsilon  \int_{0^+}^z e^{\lambda u} h(u) du \right)}{f(z)+\epsilon h(z)} - \ell (z).
\end{align}
We also compute the derivative of $p^{\epsilon}(z)$ with respect to $\epsilon$ as follows
\begin{align}\label{eq:43}
 \frac{d p^{\epsilon}(z)}{d \epsilon} \Big|_{\epsilon=0} =\delta e^{-\lambda z} \frac{\int_{0^+}^z e^{\lambda u} h(u) du}{f(z)}- \frac{h(z)}{f(z)} p(z). 
\end{align}
Based on \eqref{eq:43}, we expand $D(p^{\epsilon}(z),\kappa(z))$ to first order in $\epsilon$, i.e.,
\begin{align}
\nonumber D&(p^{\epsilon}(z),\kappa(z)) \\ 
&= D(p(z),\kappa(z)) 
+ \epsilon \frac{\partial D(p(z),\kappa(z))}{\partial p(z)} \frac{d p^{\epsilon}(z)}{d \epsilon} \Big|_{\epsilon=0}+‎ O(\epsilon^2)‎.\label{eq:42}
\end{align}
Substituting $D(p^{\epsilon}(z),\kappa(z))$ and $f^{\epsilon}(z)$ into \eqref{eq:41} results in \eqref{eq:45}-\eqref{eq:46}, where we have neglected the higher order terms of $\epsilon$ (i.e., $O(\epsilon^2)$).
\setcounter{equation}{64}
By substituting \eqref{eq:46} into \eqref{eq:38}, we establish the following necessary condition for a local and therefore a global optimal power policy $p(z)$
\begin{align}
\nonumber &\int_{0^+}^L D(p(z),\kappa(z)) \frac{h(z)}{\kappa(z)} dz -\int_{0^+}^L \frac{\partial D(p(z),\kappa(z))}{\partial p(z)}  \frac{h(z)}{\kappa(z)}p(z) dz \\
\nonumber &\hspace{2mm}+ \delta \left[ \int_{0^+}^L \frac{\partial D(p(z),\kappa(z))}{\partial p(z)} \frac{1}{\kappa(z)} \left(\int_{0^+}^z e^{-\lambda (z-u)} h(u) du \right) dz \right] \\
\label{eq:47} &= 0. 
\end{align}
By changing the order of the double integral in the third term of \eqref{eq:47} we obtain
\begin{align}
\nonumber \int_{0^+}^L  h(z) &\Bigg( \delta \int_z^L \frac{\partial D(p(u),\kappa(u))}{\partial p(u)} \frac{e^{-\lambda (u-z)}}{\kappa(u)} du +  \frac{D(p(z),\kappa(z))}{\kappa(z)} \\ 
\label{eq:48} &- \frac{\partial D(p(z),\kappa(z))}{\partial p(z)}\frac{p(z)}{\kappa(z)} \Bigg) dz=0.
\end{align} 
Equation \eqref{eq:48} holds for all perturbation functions $h(z)$, that satisfy \eqref{eq:31} and \eqref{eq:37}. We can thus rewrite \eqref{eq:48} as follows
\begin{align}
\nonumber \int_{0^+}^L \hspace{-2mm} h(z) &\Bigg( \delta \int_z^L \frac{\partial D(p(u),\kappa(u))}{\partial p(u)} \frac{e^{-\lambda (u-z)}}{\kappa(u)} du + \frac{D(p(z),\kappa(z))}{\kappa(z)} \\ 
 &- \frac{\partial D(p(z),\kappa(z))}{\partial p(z)}\frac{p(z)}{\kappa(z)}
+ \frac{C_1}{\kappa(z)} +C_2 \Bigg) dz=0, \label{eq:49}
\end{align}
where $C_1$ and $C_2$ are two free parameters. Therefore, based on the fundamental lemma of the calculus of variations we derive the following necessary condition
\begin{align}
\nonumber \delta &e^{\lambda z} \kappa(z) \int_z^L \frac{\partial D(p(u),\kappa(u))}{\partial p(u)} \frac{e^{-\lambda u}}{\kappa(u)} du + D(p(z),\kappa(z)) \\
\label{eq:50} &- \frac{\partial D(p(z),\kappa(z))}{\partial p(z)}p(z) + C_1+C_2 \kappa(z) =0,
\end{align}
where $z>0$. Eq. \eqref{eq:50} is an integro-differential equation involving $p(z)$ and $\kappa(z)$. However, by multiplying both sides of \eqref{eq:50} by $e^{-\lambda z}$ and differentiating both sides with respect to $z$ and further simplifications, we derive a first order non-linear autonomous ordinary differential equation (ODE) equivalent to \eqref{eq:50}.

We now consider the inequality in \eqref{eq:39}, where the l.h.s can be written as
\begin{align}
\nonumber \mathbb{E}_{\pi} &\left[D^{\dagger}\left(p(Z),\frac{1}{\kappa^{\epsilon}(Z)} \right) \right] \\ 
&=  \pi_0 D^{\dagger} \left( p(0),\kappa(0) \right) 
 + \int_{0^+}^L D^{\dagger}(p(z),\kappa^{\epsilon}(z)) f(z) dz.\label{eq:52}
\end{align}
We compute $\kappa^{\epsilon}(z)$ using \eqref{eq:29:1} as follows
\begin{align}\label{eq:53}
\kappa^{\epsilon}(z) = \frac{1}{\dfrac{1}{\kappa(z)}+\epsilon g(z)},
\end{align}
where its derivative with respect to $\epsilon$ at $\epsilon=0$ can easily be obtained as below
\begin{align}\label{eq:55}
\frac{d \kappa^{\epsilon}(z)}{d \epsilon} \Big|_{\epsilon=0} = - g(z) \kappa^2(z).
\end{align} 
Similar to \eqref{eq:42}, we write the Taylor series of $D(p(z),\kappa^{\epsilon}(z))$ up to first order term in $\epsilon$ as below
\begin{align}
\nonumber D&(p(z),\kappa^{\epsilon}(z)) \\
 &= D(p(z),\kappa(z))
+ \epsilon \frac{\partial D(p(z),\kappa(z))}{\partial \kappa(z)} \frac{d \kappa^{\epsilon}(z)}{d \epsilon} \Big|_{\epsilon=0}+‎ O(\epsilon^2)‎.\label{eq:54}
\end{align} 
Therefore, by substituting \eqref{eq:29:1}, \eqref{eq:55} and \eqref{eq:54} into \eqref{eq:52}, this is then reduced to \eqref{eq:57}.\setcounter{equation}{73}

By substituting \eqref{eq:57} into \eqref{eq:39} and neglecting the higher order terms of $\epsilon$, we establish the following necessary condition for a local and therefore a global optimal mismatch factor $\kappa(z)$,
\begin{align}\label{eq:58}
\int_{0^+}^L g(z) f(z) \left[ D(p(z),\kappa(z)) - \frac{\partial D(p(z),\kappa(z))}{\partial \kappa(z)} \kappa(z) \right] dz = 0.
\end{align}
Equation \eqref{eq:58} holds for all perturbation functions $g(z)$ that satisfy \eqref{eq:37:1}. Therefore, we rewrite \eqref{eq:58} as
\begin{align}\label{eq:59}
\hspace{-1mm}\int_{0^+}^L \hspace{-2mm} g(z) f(z) \hspace{-1mm} \left[ D(p(z),\kappa(z)) - \frac{\partial D(p(z),\kappa(z))}{\partial \kappa(z)} \kappa(z) + \beta \right] \hspace{-1mm} dz \hspace{-1mm}= 0,
\end{align} 
where $\beta$ is a constant. Based on the fundamental lemma of the calculus of variations, we thus derive the following necessary condition
\begin{align}\label{eq:60}
D(p(z),\kappa(z)) - \frac{\partial D(p(z),\kappa(z))}{\partial \kappa(z)} \kappa(z) + \beta = 0, \quad z>0.
\end{align}
Solutions to the equations in \eqref{eq:50} and \eqref{eq:60} determine a locally optimal power policy $p(z)$ as well as a locally optimal mismatch factor policy $\kappa(z)$.
\subsubsection{Instantaneous distortion}\label{sub:1}
In this part, we discuss an interesting consequence of \eqref{eq:60} on the instantaneous distortion $D(p(z),\kappa(z)), z>0$. Specifically, as formally presented in Lemma \ref{lemma:5}, a locally optimal solution has the property that the power policy and the mismatch factor policy are adjusted in such a way that $D(p(z),\kappa(z))$ is constant for $z>0$. In other words, if the transmission power is decreased (or increased) due to a change in the battery charge, a locally optimal mismatch factor will always dynamically adjust so that the instantaneous distortion is maintained to a constant level. As elaborated in section \ref{sec:4}, one can think of the transmission power $p(z)$ and the mismatch factor $\kappa(z)$ as limited communication resources that must maintain long-term averages, and thus the transmitter is trading off one for the other. More precisely, since the mismatch factor must satisfy \eqref{eq:20} (i.e., the inverse mismatch factor averages to one), as the availability of one communication
resource (say the transmission power) increases due to a large battery charge, the
transmitter employs a large transmission power and saves on the other communication resource (i.e.,
the mismatch factor) by then using fewer channel uses per source symbol. Likewise, when the battery charge is low, the transmitter reduces its transmission power, but employs a large mismatch factor. Formally, we have the following Lemma.
\begin{lemma}\label{lemma:5}
The constant $\beta$ in \eqref{eq:60} must be in the range $\beta_{\rm min} < \beta < \beta_{\rm max}$, where 
\begin{equation}\label{eq:62:1}
\beta_{\rm max}:= \lim_{D \downarrow 0} \hspace{3mm} \frac{R_s(D)}{R_s^{\prime}(D)},
\end{equation} 
and
\begin{equation}\label{eq:62:2}
\beta_{\rm min}:= \lim_{D \uparrow D_{\rm max}} \hspace{3mm} \frac{R_s(D)}{R_s^{\prime}(D)} - D_{\rm max}.
\end{equation}
Furthermore, for every such a choice of $\beta$, \eqref{eq:60} results in a unique constant solution for $D(p(z),\kappa(z)), \; z>0$. 
\end{lemma}
\begin{proof}
See Appendix \ref{app:3}.
\end{proof}
It is easy to verify that for both Gaussian and binary sources, $\lim_{D \downarrow 0} \hspace{3mm} \dfrac{R_s(D)}{R_s^{\prime}(D)} = 0$ and $\lim_{D \uparrow D_{\rm max}} \hspace{3mm} \dfrac{R_s(D)}{R_s^{\prime}(D)} = 0$. Therefore, for the Gaussian source we have the bound $- \sigma^2 < \beta < 0$. Likewise, for the binary source we have the bound $- \min \lbrace \mathsf{p}, 1-\mathsf{p} \rbrace < \beta < 0$.
\begin{remark}
Note that the derivation starting with \eqref{eq:52} and leading to the differential equation for the mismatch factor in \eqref{eq:60} has not used the assumption that $B(z)=\exp(-\lambda z)$. Thus, the structural result on the instantaneous distortion being constant for $z > 0$, Lemma 2 as well as the mismatch factor and instantaneous distortion in the Gaussian case to be shown in \eqref{eq:79} and \eqref{eq:83} are valid for all energy distributions $B(z)$.
\end{remark}

\subsection{A Constant Bandwidth Mismatch Factor}  
So far, we have studied a general JSCC scheme where the mismatch factor is adaptively adjusted according to the available battery charge. However it is also interesting to compare the results with the simpler scheme where the mismatch factor is fixed. Thus, we now consider the case where the mismatch factor is constant and does not adapt to the battery charge. For a fair comparison with the general case of dynamic mismatch factor, we retain the constraint \eqref{eq:25:3} which results in a constant bandwidth mismatch factor of unity, i.e., $\kappa(z)=1, \; \forall z \geq 0$. Therefore, there is only one design parameter, the transmission power $p(z)$, to minimize the total average distortion at the receiver. We thus have 
\begin{equation}
R_s(D(z))=R_c\left( p\left(z \right) \right),
\end{equation}
and thereby the distortion-rate function is computed only in terms of the transmission power $p(z)$. The optimization problem is therefore described as follows
\begin{align}\label{eq:63}
 &\hspace{-1mm}\inf_{f(z),\pi_0}  \hspace{5mm}   \pi_0 \tilde{D}(p(0)) + \int_{0^+}^L \tilde{D}\left(p(z) \right) f(z) dz\\
\label{eq:63:1}
\mbox{s.t.} &\hspace{0cm}\; f(z)\left( p(z) + \ell (z) \right)=\delta e^{-\lambda z}\Big(\pi_0 + \int_{0^+}^L e^{\lambda u}f(u)du\Big) \\ 
\label{eq:63:2}
&\pi_0+\int_{0^+}^L f(z) dz =1\\
\label{eq:63:3}
&f(z)\geq0,\quad \pi_0\geq0,
\end{align}
where $\tilde{D}(p(z)):=D(p(z),1)$ is the average distortion with bandwidth mismatch factor of unity. Moreover, the distortion lower bound in this case is the same as $D_{\rm LB}$ in \eqref{eq:26:1}. As a necessary condition for a local and thus a global optimal power policy $p(z)$, we can directly obtain the following equation by replacing $\kappa(z)=1$ into \eqref{eq:50} with the substitution $\lambda(C_1 + C_2) = C$,
\begin{equation}\label{eq:71}
\delta \int_z^L  \tilde{D}^{\prime}(p(u)) e^{-\lambda (u-z)} du + \tilde{D}(p(z)) -  \tilde{D}^{\prime}(p(z)) p(z) + C =0,
\end{equation}
where $\tilde{D}^{\prime}(\cdot)$ denotes the derivative of $\tilde{D}(\cdot)$, and $C$ is a free parameter. As explained below \eqref{eq:50}, we can derive a first order non-linear autonomous ODE equivalent to \eqref{eq:71} which is found to be
\begin{equation}\label{eq:72}
\lambda \tilde{D}(p(z)) + \left(  \delta - \lambda p(z) \right) \tilde{D}^{\prime}(p(z)) + p(z) p^{\prime}(z) \tilde{D}^{\prime\prime} (p(z))  +   C = 0,
\end{equation}
where $\tilde{D}^{\prime\prime}(\cdot)$ denotes the second derivative of $\tilde{D}(\cdot)$.
\begin{remark}
This ODE is identical in form to (102) in \cite{khuzani2013online}. Any solution $p(z)$ of the ODE in \eqref{eq:72} for $C < - \lambda \tilde{D} (\frac{\delta}{\lambda})$ is non-decreasing in $z$, for $p(z)>0$. The proof of this is similar to that of \cite[Lemma 1]{khuzani2013online} except for the change in the direction of the inequality which is due to the fact that the rate function $r(p)$ in (102) of \cite{khuzani2013online}, which is concave, is replaced with the distortion function $\tilde{D}(p)$ in \eqref{eq:72}, which is convex.
\end{remark}

\section{Gaussian Source and Channel}\label{sec:6}
In this section, we specialize our results for a Gaussian source using the rate-distortion function $R_s(D)$ given in \eqref{eq:3:1}. We assume that the Shannon rate function $R_c(p)= \dfrac{1}{2} \log_2 (1+p/N)$ is considered for the channel coding rate. Consequently, the distortion function $D(p,\kappa)$ can be computed as
\begin{equation}\label{eq:74}
D(p(z),\kappa(z))= \sigma^2 \left( 1+\frac{p(z)}{N} \right)^{-\kappa(z)},
\end{equation} 
or equivalently
\begin{equation}\label{eq:74:1}
D^{\dagger}(p(z),1/\kappa(z))= \dfrac{\sigma^2}{\kappa(z)} \left(1+\dfrac{p(z)}{N} \right)^{-\kappa(z)}.
\end{equation}
Moreover, the lower bound on the average distortion in this case is
\begin{equation}\label{eq:75}
\mathbb{E}_{\pi}\left[D^{\dagger}\left(p(Z),\frac{1}{\kappa(Z)}\right)\right] \geq \sigma^2 \left(1+\frac{\delta}{\lambda}\frac{ \left( 1- e^{-\lambda L} \right) }{N}\right)^{-1}.
\end{equation}
To determine the structure of a locally optimal scheme, we first replace $D(p(z),\kappa(z))$ in \eqref{eq:60} by its closed form expression in \eqref{eq:74} to obtain
\begin{equation}\label{eq:76}
\sigma^2 \left( 1+\frac{p(z)}{N} \right)^{-\kappa(z)}\hspace{-2mm} \left( 1+\kappa(z) \ln (1+\frac{p(z)}{N})  \right) + \beta =0, \; z>0,
\end{equation}
which can be rewritten as
\begin{align}
\nonumber \left( -1 -\kappa(z)\ln \left( 1+\frac{p(z)}{N} \right) \right) &\exp \left( \hspace{-1mm}-1 -\kappa(z) \ln \left( 1+\frac{p(z)}{N} \right) \hspace{-1mm}\right) \\
\label{eq:78}  &= \frac{\beta}{\sigma^2 e}, 
\end{align}
for $z>0$. Therefore, we have
\begin{equation}\label{eq:79}
\kappa \left( z;\beta \right)= -\frac{W_n\left( \dfrac{\beta}{\sigma^2 e} \right)+1}{\ln \left( 1+\dfrac{p(z)}{N} \right)}, \quad z>0,
\end{equation}
where $W_n(\cdot)$ denotes the $LambertW$ function \cite{corless1996lambertw} that takes either real or complex values and has an infinite number of branches, each denoted by an integer $n$. The notation $\kappa \left( z;\beta \right)$ emphasizes the dependence of $\kappa(z)$ in \eqref{eq:79} on the choice of $\beta$. Moreover, \eqref{eq:79} shows that $\kappa(z;\beta)$ is a decreasing function of $z$, whenever $p(z)$ is an increasing function of $z$ and vice versa.
\begin{remark}\label{lemma:3}
Based on the properties of the $LambertW$ function, $n=0$ and $n=-1$ are the only branches that yield a real value for $W_n(x)$, where $W_0(x)\geq -1$ for $x \geq -1/e$ and $W_{-1}(x)<-1$ for $-1/e < x < 0$. In addition, we require that $W_n(\dfrac{\beta}{\sigma^2 e}) < -1$ in order to have $\kappa(z;\beta) >0$ in \eqref{eq:79}. Therefore, $n=-1$ is the only acceptable branch that results in real positive values of $\kappa(z;\beta)$. This provides another proof for the fact that $-1/e < \dfrac{\beta}{\sigma^2 e} < 0$ or equivalently $-\sigma^2 < \beta <0$ (see Lemma \ref{lemma:5}). 
\end{remark}
From \eqref{eq:79}, for every fixed value of $\beta$, we have
\begin{align}
\nonumber -\frac{1}{2} &\left( W_{-1}  \left( \frac{\beta}{\sigma^2 e} \right) +1 \right) \log_2 e \\ 
 &= \kappa \left( z; \beta \right) \times  \frac{1}{2}\log_2 \left( 1+\frac{p(z)}{N} \right) , \quad z>0. \label{eq:80}
\end{align}
On the other hand, this choice of $\beta$ results in the instantaneous distortion $D(z):=D(p(z),\kappa(z))$ in \eqref{eq:4} as follows
\begin{equation}\label{eq:81}
\frac{1}{2} \log_2 \frac{\sigma^2}{D(z)}= \kappa(z;\beta) \times \frac{1}{2}\log_2 \left( 1+\frac{p(z)}{N} \right) , \quad z\geq 0 .
\end{equation}
Therefore, equating the l.h.s of \eqref{eq:80} and \eqref{eq:81} yields
\begin{equation}\label{eq:82}
\ln \frac{\sigma^2}{D(z)} = - \left( W_{-1} \left( \frac{\beta}{\sigma^2 e} \right)+1 \right), \quad z>0.
\end{equation}
Let $\beta^{\star}$ denote the optimal value of the free parameter $\beta$. Thus, \eqref{eq:82} results in the following associated optimal instantaneous distortion $D^{\star}(z)$ 
\begin{equation}\label{eq:83}
D^{\star}(z) = \left\{ 
  \begin{array}{l l}
    \sigma^2 \exp \left( W_{-1} \left( \dfrac{\beta^{\star}}{\sigma^2 e} \right) +1 \right)\; & \quad z>0 \\
    \sigma^2 \; & \quad z=0.
  \end{array} \right.
\end{equation}
As elaborated in Section \ref{sub:1} and Appendix \ref{app:3}, \eqref{eq:83} is the unique constant solution of \eqref{eq:60} for $z>0$ in the Gaussian case. Specifically, $p(z)$ and $\kappa(z)$ are jointly optimized such that when using a high (respectively low) transmission power, the transmitter uses a low (respectively high) mismatch factor to maintain the optimal instantaneous distortion to a constant value. 

\section{Numerical Results}\label{sec:8}
In this section, we consider numerical solutions to \eqref{eq:50} and \eqref{eq:60} to obtain an efficient power policy $p(z)$ and mismatch factor policy $\kappa(z)$. In particular, although there is no apparent closed-form solution to \eqref{eq:50} for $p(z)$, for every choice of the constants $C_1$, $C_2$ and the initial condition $p(0^+)$, we can apply numerical ODE solution methods. More precisely, given a fixed choice of $\beta$, from \eqref{eq:60} we can in principle solve for $\kappa(z)$ in terms of $p(z)$ (although other than the Gaussian case where a closed-form is found in \eqref{eq:79}, this must be done numerically). We then substitute $\kappa(z)$ thus computed into \eqref{eq:50} and obtain an ODE for $p(z)$ in terms of $C_1$, $C_2$ and $\beta$ that can be solved numerically. Once $p(z)$ is thus found, one can then directly obtain $\kappa(z)$ using \eqref{eq:60} again.

We obtain from \eqref{eq:16} that
\begin{equation}\label{eq:115}
\frac{f(z)e^{\lambda z}}{\pi_0+ \int_{0^+}^z e^{\lambda u}f(u) du} = \frac{\delta}{p(z)+\ell(z)},
\end{equation}
where by integrating both sides over $(0^+,z]$ and performing some simplifications, this becomes
\begin{equation}\label{eq:116}
\pi_0 + \int_{0^+}^z e^{\lambda u} f(u) du = \pi_0 \exp\left( \int_{0^+}^z \frac{\delta}{p(u)+\ell(u)} du \right).
\end{equation}
Taking the derivative of both sides in \eqref{eq:116} with respect to $z$, we compute the density $f(z)$, provided that $\pi_0$ is known, as follows
\begin{equation}\label{eq:117}
f(z) = \pi_0 \frac{\delta e^{-\lambda z}}{p(z)+\ell(z)} \exp\left( \int_{0^+}^z \frac{\delta}{p(u)+\ell(u)} du \right) .
\end{equation}
By combining \eqref{eq:117} and \eqref{eq:25:2}, we compute the atom $\pi_0$ as below
\begin{equation}\label{eq:118}
\pi_0 = \left( 1+ \int_{0^+}^L \frac{\delta e^{-\lambda z}}{p(z)+\ell(z)} \exp\left( \int_{0^+}^z \frac{\delta}{p(u)+\ell(u)} du \right) dz \right)^{-1}\hspace{-3mm}.
\end{equation}
Moreover, the value of the mismatch factor $\kappa(0):=\kappa(z)\vert_{z=0}$ when the battery is exhausted is obtained from \eqref{eq:25:3} as
\begin{equation}\label{eq:119}
\kappa(0) = \frac{\pi_0}{1-\int_{0^+}^L f(z)/ \kappa(z;\beta)  dz}.
\end{equation}

We now study a single-user EH communication system with energy arrival rates $\delta=\lambda=1$, and noise power $N=1$. To find locally optimal policies for $p(z)$ and $\kappa(z)$, one needs to search for optimized values of the free constants $\beta,\; C_1,$ and $C_2$. In the following we separately investigate the cases of Gaussian and binary sources.
\subsection{A Gaussian Source over a Gaussian Channel}
For a standard Gaussian source $\mathcal{N}(0,1)$, we have the bound $-1 < \beta < 0$ from Section \ref{sub:1}. We now examine two cases of leakage: $(i)$ zero leakage rate $\ell(z)=0$ for an ideal battery and $(ii)$ increasing leakage rate $\ell(z)=1-e^{-z}$ for an imperfect battery. Obviously, our analysis is general enough to allow us to study different leakage behaviours $\ell(z)$. However as stated in \cite{devillers2012general}, \textit{``batteries leak most right after being charged"}, and examples of rechargeable batteries with leakage (or self-discharge) rates that increase monotonically with the battery charge are nickel-cadmium and nickel-hydrogen cells \cite{zimmerman2004self}.

We also consider the battery capacities $L=1, 2, ..., 5$. Table \ref{tab:1} and Table \ref{tab:2} show the total average distortion $D_{\rm avg}$, the distortion lower bound $D_{\rm LB}$, and good values of the constants $\beta,\; C_1,$ and $C_2$ found by numerical search for both cases. Furthermore, the initial condition of the ODE for $p(z)$ was chosen to be $p(0^+)=0.001$. This choice of $p(0^+)$ is justified by the fact that a small amount of available energy in the battery should entail a small transmission power, as otherwise, the battery will be completely depleted before the next energy arrival. Numerical simulations have shown that if $p(0^+)$ is chosen to be sufficiently small, it does not significantly change the distortion performance. A similar observation was made in \cite{khuzani2013online}.
\begin{table*}[t!]
\centering
\caption{\label{tab:1}\small Distortion lower bound $D_{\rm LB}$, average distortion $D_{\rm avg}$, and good values of the constants $C_1$, $C_2$ and $\beta$ of a standard Gaussian source $\mathcal{N}(0,1)$, for different battery capacities with the initial condition $p(0^+)=0.001$, when $\ell(z)=0$. The ratio of $D_{\rm avg}$ to $D_{\rm LB}$ quantifies the gap between the average distortion and the lower bound.}
  \begin{tabular}[!t]{ccccc}
\hline\hline
Capacity of the Battery & $D_{\rm LB}$ & $D_{\rm avg}$ & $D_{\rm avg}/D_{\rm LB}$ &  Constants \\
\hline
$L=1$ & $0.6127$  & $0.6971$ & $1.13$& $\beta=-0.9485, C_1=-0.95, C_2=0.24$\\
 \hline
$L=2$ &  $0.5363$ &$0.6147$ & $1.14$ & $\beta=-0.9137, C_1=-0.92, C_2=0.30$\\
  \hline
$L=3$ & $0.5128$&$0.5765$ & $1.12$& $\beta=-0.8940, C_1=-0.90, C_2=0.32$ \\ 
  \hline
$L=4$ &  $0.5046$ & $0.5559$ &$1.10$& $\beta=-0.8822, C_1=-0.89, C_2=0.32$ \\  
  \hline
$L=5$  &  $0.5017$ & $0.5417$  & $1.07$& $\beta=-0.8738, C_1=-0.89, C_2=0.34$ \\
\hline
\end{tabular}
\end{table*}
\begin{table*}[t!]
\centering
\caption{\label{tab:2}\small Distortion lower bound $D_{\rm LB}$, average distortion $D_{\rm avg}$, and good values of the constants $C_1$, $C_2$ and $\beta$ of a standard Gaussian source $\mathcal{N}(0,1)$, for different battery capacities with the initial condition $p(0^+)=0.001$, when $\ell(z)=1-e^{-z}$. The ratio of $D_{\rm avg}$ to $D_{\rm LB}$ quantifies the gap between the average distortion and the lower bound.}
  \begin{tabular}[!t]{ccccc}
\hline\hline
Capacity of the Battery & $D_{\rm LB}$ &$D_{\rm avg}$ & $D_{\rm avg}/D_{\rm LB}$& Constants \\
\hline
$L=1$ &  $0.6127$& $0.7445$  & $1.21$&$\beta=-0.9641, C_1=-0.97, C_2=0.06$ \\
 \hline
$L=2$ &  $0.5363$ & $0.6876$  & $1.28$ &$\beta=-0.9450, C_1=-0.95, C_2=0.17$\\
  \hline
$L=3$ & $0.5128$ & $0.6659$ & $1.29$& $\beta=-0.9366, C_1=-0.94, C_2=0.29$  \\ 
  \hline
$L=4$ & $0.5046$  & $0.6596$ &  $1.30$ &$\beta=-0.9300, C_1=-0.93, C_2=0.32$ \\  
  \hline
$L=5$  & $0.5017$ &$0.6566$  &$1.30$& $\beta=-0.9302, C_1=-0.93, C_2=0.34$ \\
\hline
\end{tabular}
\end{table*}

It is evident from Table \ref{tab:1} that as the battery capacity increases, the achieved distortion decreases. In particular, for $L=5$ when the leakage rate is zero this scheme can achieve a distortion that is at most $7\%$ above the lower bound. In fact, for an ideal battery with infinite capacity as discussed below \eqref{eq:95-1}, the lower bound is asymptotically tight and can be approximated arbitrarily well with a constant transmission power policy and a constant mismatch factor policy. Moreover, for case $(ii)$ with non-zero leakage, the achieved average distortions are larger compared to that of case $(i)$. Note that the distortion lower bound in Table \ref{tab:2} is relatively more loose, since it does not depend on the leakage rate (see \eqref{eq:28:2}), while for the case of a non-zero leakage the optimal performance should depend on the leakage model. Therefore, we do not expect the lower bound to be asymptotically achievable, as it is without leakage. As the leakage model has increasing leakage with increasing battery charge, it is not unreasonable for the lower bound to become more loose with larger battery capacities. Nevertheless, the ratio in Table \ref{tab:2} must saturate to a finite value as $L \to \infty$. This is because $i)$ the distortion lower bound for an infinite battery capacity is $D^{\dagger}\left (\delta/\lambda, 1 \right)$ which is strictly positive in this case, and $ii)$ the achieved distortion will converge to some value, and hence, so will the ratio.

Fig. \ref{fig:b} shows the transmission power policy $p(z)$, for the case of $L=5$, for an ideal battery as well as an imperfect battery with three different leakage rates, i.e., increasing $\ell(z)=1-e^{-z}$, decreasing $\ell(z)=e^{-z}$, and constant $\ell(z)=1$. 
 
\begin{figure}[t!]
\centering
\epsfig{width=9cm, figure=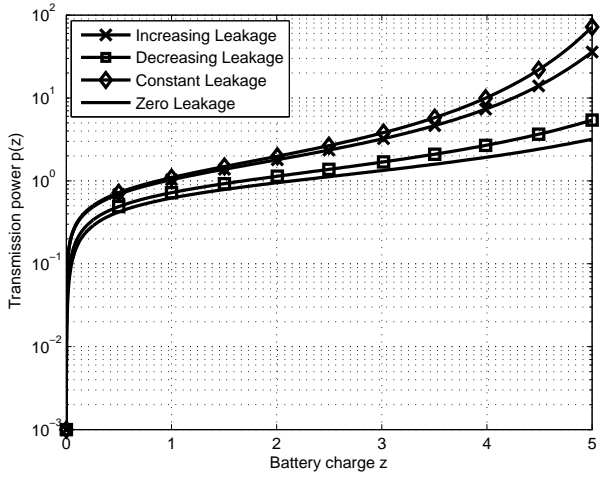}
\caption{\small Power policy $p(z)$ under 4 different scenarios, namely, zero leakage, increasing and decreasing leakage and constant leakage, when $p(0^+)=0.001$.}
\label{fig:b}
\end{figure}
We observe that in all cases, the designed transmission power monotonically increases as the battery charge increases. This is due to the fact that when the remaining charge in the battery is close to the capacity limit, new energy arrivals are likely to make the battery overflow. Therefore, the transmitter consumes a large transmission power in order to avoid losing energy. Interestingly, for the increasing leakage and the constant leakage cases, the allocated transmission power increases faster compared to the transmission power for an ideal battery with the same battery charge. This result is intuitive, since an efficient transmission power policy mitigates the potentially large energy loss due to leakage by rapidly consuming the stored energy before it is lost. Figs. \ref{fig:c} and \ref{fig:d} illustrate the corresponding mismatch factor $\kappa(z)$ and absolutely continuous part of the density function of the available charge in the battery, respectively. It can be seen that as the energy in the battery decreases, the mismatch factor increases. In other words, the low transmission power due to reduced charge in the battery is compensated by using longer channel codewords. Conversely, when the transmission power is large, codewords of smaller length are used so that the constraint in \eqref{eq:20} is satisfied. As expected, an increasing transmission power results in a decreasing density function $f(z)$. In other words, the storage process spends a smaller fraction of time at higher battery charges that have larger transmission power.
\begin{figure}[t!]
\centering
\epsfig{width=9cm, figure=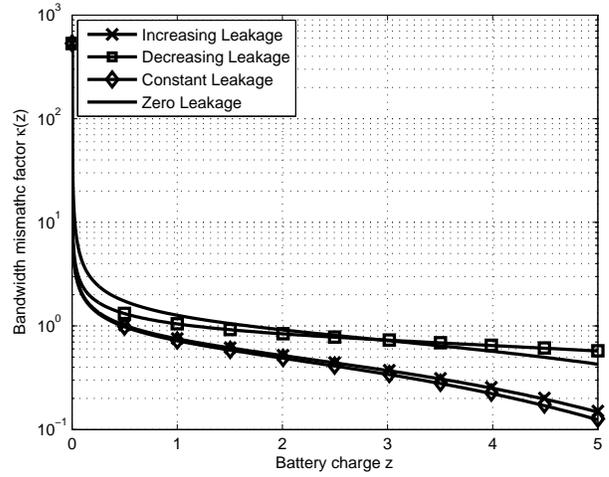}
\caption{\small Bandwidth mismatch factor $\kappa(z)$ under 4 different scenarios, namely, zero leakage, increasing and decreasing leakage and constant leakage, when $p(0^+)=0.001$.}
\label{fig:c}
\end{figure}
\begin{figure}[t!]
\centering
\epsfig{width=9cm, figure=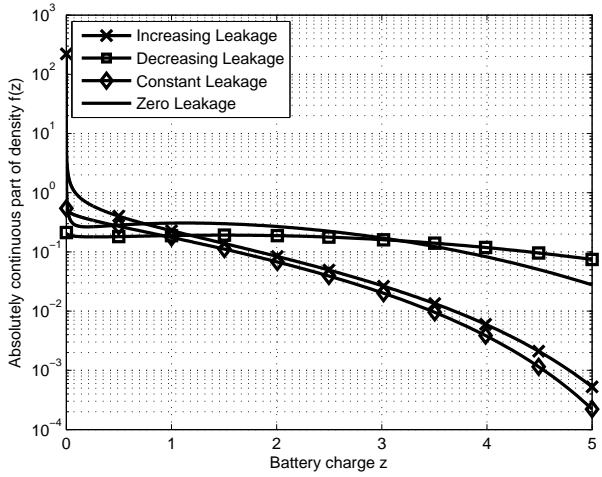}
\caption{\small Density function $f(z)$ under 4 different scenarios, namely, zero leakage, increasing and decreasing leakage and constant leakage, when $p(0^+)=0.001$.}
\label{fig:d}
\end{figure}

The average distortion for the two cases of dynamically varying bandwidth mismatch factor and constant bandwidth mismatch factor is illustrated in Fig \ref{fig:a}, both as functions of the battery capacity. It can be observed that for a battery with capacity in the range $2 \leq L \leq 12$, a communication system with an adaptive mismatch factor, as proposed in this paper, performs better compared to that of $\kappa(z)=1$. Whether this gap is considered small or not depends on the application. If it is negligible for some applications, then using a fixed channel coding scheme would result in a reduced complexity system with negligible impact on performance. On the other hand, regardless of whether we are using a dynamic mismatch factor or a constant mismatch factor, for large values of $L$ the average distortion of both coding schemes approach the lower bound and merge asymptotically. Likewise, for small battery capacities the average distortion of both coding schemes approach $D_{\rm max}$ and merge as well.
\begin{figure}[t]
\centering
\epsfig{width=9cm, figure=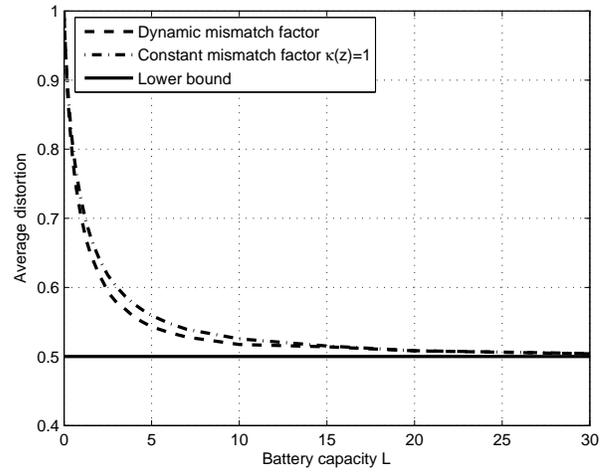}
\caption{\small Average distortion for the two cases of adaptive mismatch factor and constant mismatch factor $\kappa(z)=1$, when the battery capacity varies in the range $0 \leq L \leq 30$.}
\label{fig:a}
\end{figure}
\subsection{A Binary Source over a Gaussian Channel}
We now consider a binary source with the Bernoulli$(1/2)$ distribution for which we have $H(\mathsf{p})=1$, and the bound $-0.5 < \beta < 0$ as explained in subsection \ref{sub:1}. Analogous to the Gaussian source, with different battery capacities $L=1, 2, ..., 5$ and for the two cases of $(i)$ zero leakage rate $\ell(z)=0$ and $(ii)$ increasing leakage rate $\ell(z)=1-e^{-z}$, we evaluate the achieved average distortion $D_{\rm avg}$, the distortion lower bound $D_{\rm LB}$ and good values of the constants found by numerical search. The results of case $(i)$ and $(ii)$ which are summarized in Tables \ref{tab:3} and \ref{tab:4}, respectively, have the same trends as in Tables \ref{tab:1} and \ref{tab:2} with increasing the battery capacity. Moreover, the numerical results further show that for $L=5$ the proposed achievable scheme with no leakage can achieve a distortion which is at most $19\%$ above the lower bound.   
\begin{table*}[t!]
\centering
\caption{\label{tab:3}\small Distortion lower bound $D_{\rm LB}$, average distortion $D_{\rm avg}$, and good values of the constants $C_1$, $C_2$ and $\beta$ of a Bernoulli$(1/2)$ source, for different battery capacities with the initial condition $p(0^+)=0.001$, when $\ell(z)=0$. The ratio of $D_{\rm avg}$ to $D_{\rm LB}$ quantifies the gap between the average distortion and the lower bound.}
  \begin{tabular}[!t]{ccccc}
\hline\hline
Capacity of the Battery &  $D_{\rm LB}$ & $D_{\rm avg}$ & $D_{\rm avg}/D_{\rm LB}$ &   Constants \\
\hline
$L=1$ &  $0.1651$ & $0.2097$ & $1.27$ & $\beta=-0.3450, C_1=-0.36, C_2=0.13$\\
 \hline
$L=2$ &  $0.1270$ &  $0.1663$ & $1.30$ & $\beta=-0.3170, C_1=-0.32, C_2=0.15$\\
  \hline
$L=3$ & $0.1161$& $0.1473$ & $1.26$ & $\beta=-0.3039, C_1=-0.31, C_2=0.16$ \\ 
  \hline
$L=4$ &  $0.1122$ & $0.1367$ & $1.21$ &  $\beta=-0.2962, C_1=-0.30, C_2=0.16$ \\  
  \hline
$L=5$  &  $0.1108$  & $0.1321$ & $1.19$ & $\beta=-0.2901, C_1=-0.29, C_2=0.16$ \\
\hline
\end{tabular}
\end{table*}
\begin{table*}[t!]
\centering
\caption{\label{tab:4}\small Distortion lower bound $D_{\rm LB}$, average distortion $D_{\rm avg}$, and good values of the constants $C_1$, $C_2$ and $\beta$ of a Bernoulli$(1/2)$ source, for different battery capacities with the initial condition $p(0^+)=0.001$, when $\ell(z)=1-e^{-z}$. The ratio of $D_{\rm avg}$ to $D_{\rm LB}$ quantifies the gap between the average distortion and the lower bound.}
  \begin{tabular}[!t]{ccccc}
\hline\hline
Capacity of the Battery &  $D_{\rm LB}$ & $D_{\rm avg}$ & $D_{\rm avg}/D_{\rm LB}$ &   Constants \\
\hline
$L=1$ &  $0.1651$ & $0.2356$ & $1.42$ & $\beta=-0.3600, C_1=-0.36, C_2=0.04$\\
 \hline
$L=2$ &  $0.1270$ &  $0.2044$ & $1.60$ & $\beta=-0.3417, C_1=-0.35, C_2=0.11$\\
  \hline
$L=3$ & $0.1161$& $0.1935$ & $1.66$ & $\beta=-0.3347, C_1=-0.34, C_2=0.15$ \\ 
  \hline
$L=4$ &  $0.1122$ & $0.1905$ & $1.69$ &  $\beta=-0.3328, C_1=-0.34, C_2=0.17$ \\  
  \hline
$L=5$  &  $0.1108$  & $0.1885$ & $1.70$ & $\beta=-0.3301, C_1=-0.33, C_2=0.18$ \\
\hline
\end{tabular}
\end{table*}

\section{Conclusion}\label{sec:9}
We have investigated the problem of joint source-channel coding in a point-to-point channel with an energy harvesting transmitter. We used a calculus of variations technique to characterize an achievable
joint source-channel coding scheme as well as an achievable transmission power policy
to minimize the distortion at the receiver. We also obtained a distortion lower bound, where we used the convexity of the distortion function and an upper bound on the average transmission power. 

For a moderate-size battery capacity, we numerically showed that the achievable distortion
with a dynamically varying bandwidth mismatch factor is smaller than that of a constant mismatch factor when the battery has no leakage. Moreover, we observed that as the battery capacity tends to infinity the achievable distortion for both coding schemes approached the lower bound. Furthermore, with a constant mismatch factor $\kappa(z)=1$, we found a constant transmission power policy that can arbitrarily approach this lower bound for an infinite battery capacity. 

As examples of continuous and discrete alphabet sources, we considered both Gaussian and binary sources to validate our analytical findings. In both cases, we showed numerically that
a good transmission power policy increases as the battery charge increases. In contrast, a good mismatch factor policy, which measures the ratio of the length of channel codewords per source symbol, is a decreasing function of the battery charge. We further examined these policies under different possibilities for battery leakage rate, i.e., zero leakage rate as well as non-zero arbitrary leakage rate.

One possible future extension is to consider a scenario where lossy source-channel communication is carried out over a time-varying channel with slow fading. In particular, the transmission power and mismatch factor in this case need to be adapted to both the battery charge and the channel state. However, similar to \cite{khuzani2013optimal}, we expect a system of coupled ODEs with one ODE per channel state, which can be challenging to solve, even numerically. Other future work includes considering other and more general energy arrival models.

\appendices
\section{}
To prove the convexity of the distortion function $D^{\dagger}(p,q)$ over the domain $p \geq 0, q>0$ where $q:=1/\kappa$, we consider two cases: ($\rm a$) $R_s^{\rm th}=\infty$, ($\rm b$) $R_s^{\rm th}<\infty$. As already discussed, an example of the former is a Gaussian source and an example of the latter is a binary source. We first prove the joint convexity with respect to $p$ and $q$ for case ($\rm a$). To do so, we compute the Hessian matrix of $D^{\dagger}(p,q)$ for $p>0, q>0$, denoted by $H_D$, and show that it is positive definite. The Hessian is given by
$$H_D = \begin{pmatrix}
\dfrac{\partial^2 D^{\dagger}}{\partial p^2} & \dfrac{\partial^2 D^{\dagger}}{\partial p \; \partial  q}  \\
\dfrac{\partial^2 D^{\dagger}}{\partial q \; \partial p} & \dfrac{\partial^2 D^{\dagger}}{\partial  q^2}
\end{pmatrix},$$
where by a simple calculation we obtain
\begin{align}
\nonumber \dfrac{\partial^2 D^{\dagger}}{\partial p^2}
\nonumber &= \hspace{2mm} \frac{R_c^{\prime\prime}(p)}{R_s^{\prime} \left( D(p,\kappa) \right)} + \left( R_c^{\prime}(p) \right)^2 \kappa \times \frac{- R_s^{\prime\prime} \left( D(p,\kappa)\right)}{\left( R_s^{\prime}(D(p,\kappa)) \right)^3 } \\ 
\nonumber &\stackrel{(\rm a)}{>} \hspace{2mm} 0,
\end{align}
where $(\rm a)$ follows from the conditions [S2] and [C2] on $R_s(D)$ and $R_c(p)$, respectively. Moreover, the strict inequality is justified by the fact that since $R_s^{\rm th}=\infty$, the distortion $D$ is strictly positive and therefore $R_s^{\prime}(D)$ is finite due to [S3]. Similarly, we can show that the determinant of $H_D$ is strictly positive, i.e.,
\begin{align}
\nonumber  \det(H_D)
\nonumber =& \hspace{0mm}\left( R_c(p) \right)^2 \kappa^3 R_c^{\prime\prime} (p) \\
\nonumber &\times \frac{1}{R_s^{\prime}(D(p,\kappa))} \times \frac{-R_s^{\prime\prime}(D(p,\kappa))}{R_s^{\prime}(D(p,\kappa))} \stackrel{(\rm b)}{>} \hspace{2mm} 0,
\end{align} 
where $(\rm b)$ again follows from conditions [S2], [S3] and [C2]. Since, $H_{D_{11}} > 0$ and $\det(H_D)> 0$, by Sylvester's criterion the matrix $H_D$ is positive definite, and it thus follows that $D^{\dagger}(p,q)$ is jointly convex over the pair $p$ and $q$. 

\begin{figure*}[t!] 
\centering  
\captionsetup{justification=centering}
\def\svgwidth{100pt} 
\input{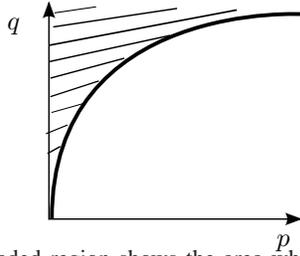} 
\caption{\small The shaded region shows the area where $q > R_c(p)/R_s^{\rm th}$.\vspace{1cm}}
\label{fig:2} 
\end{figure*} 
To prove the joint convexity with respect to $p$ and $q$ for case ($\rm b$) where $R_s^{\rm th}<\infty$, we recall that $D^{\dagger}(p,q)=0$ for $ R_c(p) \geq q R_s^{\rm th}$. Although the function $D^{\dagger}(p,q)$ is continuous everywhere, and in particular at the points where $\kappa R_c(p)=R_s^{\rm th}$, the second derivative at these points may not necessarily exist and a more complicated analysis is required. Therefore, as illustrated in Fig. \ref{fig:2} we separate the region $p>0$ and $q >0$ into two parts: the open shaded region, $ R_c(p) < q R_s^{\rm th}$, over which the convexity argument reduces to the case ($\rm a$) and the closed unshaded region, $ R_c(p) \geq q R_s^{\rm th}$, over which $D^{\dagger}(p,q)=0$. It is not hard to see that the unshaded region given by $R_c(p) \geq q R_s^{\rm th}$ (or equivalently $q \leq R_c(p)/R_s^{\rm th}$) is convex. This is due to the fact that $R_c(p)$ is a concave function and thus the region it traces is convex. Now, consider two arbitrary points $\alpha_1=(p_1,q_1)$ and $\alpha_2=(p_2,q_2)$ such that $\alpha_1, \alpha_2 \in \lbrace (p,q): p>0, q>0 \rbrace$. For $\lambda \in [0, 1]$, we define the function $g(\lambda)$ as follows
\begin{align}
\nonumber g(\lambda) :=& \hspace{2mm} D^{\dagger}\left( \lambda \alpha_1 + (1-\lambda) \alpha_2 \right) \\
\nonumber =& \hspace{2mm}D^{\dagger} (\lambda p_1 + (1-\lambda)p_2, \lambda q_1 + (1-\lambda) q_2).
\end{align}
\begin{figure*}[t!]
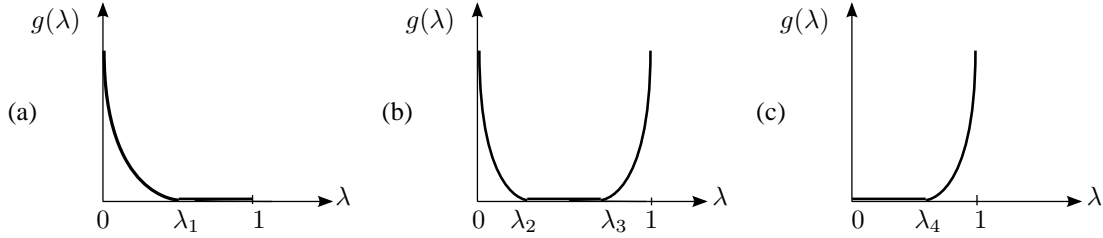

\centering
\begin{subfigure}
\centering 
\def\svgwidth{90pt} 
\input{Figure7.eps_tex}\hspace{1.2cm}
\end{subfigure}
~
\begin{subfigure}
\centering 
\def\svgwidth{90pt} 
\input{Figure8.eps_tex}\hspace{1.2cm}
\end{subfigure}
~
\begin{subfigure}
\centering 
\def\svgwidth{90pt} 
\input{Figure9.eps_tex}\hspace{1.4cm}
\end{subfigure}
\caption{\small Special cases of the function $g(\lambda)$.}
\label{fig:3}
\end{figure*}
If $g(\lambda)$ is convex for all choices of $\alpha_1$ and $\alpha_2$, so is $D^{\dagger}(p,q)$.
With respect to the closed line segment connecting these two points, i.e., $\mathcal{L} = \lbrace \lambda \alpha_1 + (1-\lambda) \alpha_2: \lambda \in [0,1] \rbrace$, the following different cases may happen:
\begin{itemize} 
\item If $\alpha_1$ and $\alpha_2$ are both in the unshaded region, the line segment $\mathcal{L}$ only passes through the unshaded region as the region is convex and thus $g(\lambda)=0$, which is convex.
\item If $\alpha_1$ and $\alpha_2$ are in two different regions (say $\alpha_1$ is in the unshaded region and $\alpha_2$ is in the shaded region), there exists $\lambda_1 \in (0,1]$ as shown in Fig. \ref{fig:3}a such that $\mathcal{L}$ lies in the shaded region for $\lambda < \lambda_1$ and it enters the unshaded region for $\lambda \geq \lambda_1$. Furthermore, once the line segment enters the convex unshaded region, it does not exit. Here, $g(\lambda)$ is continuous for $\lambda \in [0,1]$, non-negative and strictly convex for $\lambda \in [0, \lambda_1)$, while $g(\lambda)=0$ for $\lambda \in [\lambda_1, 1]$. Thus, $g(\lambda)$ is convex.
\item If $\alpha_1$ and $\alpha_2$ are both in the shaded region, then either $\mathcal{L}$ only passes through the shaded region where the convexity of $g(\lambda)$ reduces to the case $(\rm a)$, or due to the convexity of the unshaded region it enters the unshaded region for one and only one contiguous closed interval $[\lambda_2,\lambda_3] \subset (0,1)$ and again returns to the shaded region for $\lambda > \lambda_3$. The function $g(\lambda)$ in the latter case is again continuous for $\lambda \in [0,1]$, non-negative and strictly convex for $\lambda \in [0, \lambda_2) \cup (\lambda_3, 1]$, while $g(\lambda)=0$ for $\lambda \in [\lambda_2, \lambda_3]$. Thus, $g(\lambda)$ is convex and this case is illustrated in Fig. \ref{fig:3}b. 
\item Similar to the second case, when $\alpha_1$ is in the shaded region and $\alpha_2$ is in the unshaded region, we have Fig. \ref{fig:3}c, where $g(\lambda)$ is convex.
\end{itemize}
Consequently, even for sources for which zero distortion can be achieved (e.g., the binary source), the function $D^{\dagger}(p,q)$ is jointly convex over the pair $p$ and $q$.

\section{}\label{app:3}
We first rewrite \eqref{eq:4} for $z>0$ as
\begin{equation}\label{eq:60:1}
R_s(D(p(z),\kappa(z)))=\kappa(z) R_c(p(z)),
\end{equation}
where by taking the first derivative of both sides with respect to $\kappa(z)$ at a fixed $z>0$ we obtain
\begin{equation}\label{eq:60:2}
\frac{\partial D(p(z),\kappa(z))}{\partial \kappa(z)} R_s^{\prime}\left( D\left( p(z),\kappa(z) \right) \right) = R_c(p(z)),
\end{equation}
or equivalently
\begin{equation}\label{eq:60:3}
\frac{\partial D(p(z),\kappa(z))}{\partial \kappa(z)}= \frac{R_c(p(z))}{R_s^{\prime}\left( D\left( p(z),\kappa(z) \right) \right)}.
\end{equation}
To simplify the notation, we fix $z>0$, and simply write $D$ in place of $D(p(z),\kappa(z))$. Thus, with the substitution \eqref{eq:60:3}, \eqref{eq:60} reduces to
\begin{equation}\label{eq:60:4}
D - \frac{R_s(D)}{R_s^{\prime}\left( D \right)} + \beta = 0,
\end{equation}
where for the second term we also used the substitution \eqref{eq:60:1}. Therefore, $D$ must be a root of \eqref{eq:60:4}. Solving \eqref{eq:60:4} for $\beta$, we then obtain
\begin{equation}\label{eq:60:5}
\beta= \frac{R_s(D)}{R_s^{\prime}\left( D \right)} - D.
\end{equation} 
We next show that the r.h.s of \eqref{eq:60:5} is strictly decreasing with respect to $D$. To do so, we take the first derivative of the r.h.s in \eqref{eq:60:5} with respect to $D$, and show that it is always negative in the open interval $0 < D < D_{\rm max}$, i.e.,
\begin{align}\label{eq:60:6}
\frac{d }{d D} \left[ \frac{R_s(D)}{R_s^{\prime}\left( D \right)} - D \right] &= \frac{-R_s^{\prime\prime}(D)R_s(D)}{\left(R_s^{\prime}(D)\right)^2} \stackrel{(\rm c)}{<} 0,
\end{align}
where $(\rm c)$ follows from the assumption of $R_s^{\prime\prime}(D)>0$ and further recalling that $R_s(D)>0$ and $R_s^{\prime}(D)$ is finite, for $0 < D < D_{\rm max}$. Therefore, for every fixed $\beta$, there is at most one real root $D$ that solves \eqref{eq:60:5}, and it does not depend on $z$. This is also true for the optimized value of $\beta^{\star}$, and therefore the associated instantaneous distortion $D^{\star}(p(z),\kappa(z))$ is constant for $z>0$. Furthermore, for there to be at least one real root, $\beta$ in \eqref{eq:60:5} must be in the range $\beta_{\rm min} < \beta < \beta_{\rm max}$, where
\begin{align}\label{eq:60:8}
\beta_{\rm max} :=& \hspace{2mm} \sup_{0 < D < D_{\rm max}}  \left[ \frac{R_s(D)}{R_s^{\prime}(D)} - D \right] \\ 
=& \hspace{2mm} \lim_{D \downarrow 0} \hspace{3mm} \frac{R_s(D)}{R_s^{\prime}(D)},
\end{align} 
and
\begin{align}\label{eq:60:9}
\beta_{\rm min} :=& \hspace{2mm} \inf_{0 < D < D_{\rm max}} \left[ \frac{R_s (D)}{R_s^{\prime}(D)} - D \right] \\ 
=& \hspace{2mm} \lim_{D \uparrow D_{\rm max}} \hspace{3mm} \frac{R_s(D)}{R_s^{\prime}(D)} - D_{\rm max}.
\end{align}

\end{document}